\pdfoutput=1
\let\TeXyear\year
\documentclass{ieeeaccess}

\let\year\TeXyear
\usepackage{tcolorbox}
\NewSpotColorSpace{PANTONE}
\AddSpotColor{PANTONE} {PANTONE3015C} {PANTONE\SpotSpace 3015\SpotSpace C} {1 0.3 0 0.2}
\SetPageColorSpace{PANTONE}
\definecolor{accessblue}{cmyk}{1, 0.3, 0, 0.2}
\definecolor{greycolor}{cmyk}{0,0,0,.8}
\usepackage{cite}
\usepackage{amsmath,amssymb,amsfonts}
\usepackage{algorithmic}
\usepackage{graphicx}
\usepackage{textcomp}
\usepackage[dvipsnames]{xcolor}  
\usepackage{tikz}
\usetikzlibrary{arrows.meta, positioning, fit, backgrounds}
\usepackage{soul}
\usepackage{mdsymbol}
\usepackage[compatibility=false]{caption}

\usepackage{url}
\usepackage{float}

\usepackage[utf8]{inputenc}
\usepackage{booktabs}
\usepackage{subcaption}

\usepackage{amsthm}
\usepackage[linesnumbered,ruled,vlined]{algorithm2e}
\usetikzlibrary{calc, arrows.meta, positioning}
\tikzstyle{blackdot} = [circle, fill=black, minimum size=2.0mm, inner sep=0pt]
\tikzstyle{whitedot} = [circle, draw=black, fill=black, minimum size=2.0mm, inner sep=0pt]

\tikzset{
  solidEdge/.style={thick, solid},
  dottedEdge/.style={thick, dashed}
}

\def\BibTeX{{\rm B\kern-.05em{\sc i\kern-.025em b}%
\kern-.08em T\kern-.1667em\lower.7ex\hbox{E}\kern-.125emX}}

\usepackage{pgfplots}
\pgfplotsset{compat=1.17} 
\usepackage[compatibility=false]{caption}      
\usepackage{graphicx}
\usepackage{siunitx} 

\newtheorem{definition}{Definition}

\newtheorem{lemma}{Lemma}
\theoremstyle{definition}
\newtheorem*{def*}{Definition}

\newcommand{\bbtwok}{\textsc{BB2K}\xspace}
\newcommand{\balancebutterfly}{\mathbin{\rotatebox[origin=c]{90}{$\hourglass$}^\star}}

\begin{document}

\history{Date of publication xxxx 00, 0000, date of current version xxxx 00, 0000.}
\doi{10.1109/ACCESS.2017.DOI}

\title{Multi-core \& GPU-based Balanced Butterfly Counting in Signed Bipartite Graphs}

\author{
\uppercase{Kiran Mekala}\authorrefmark{1}, 
\uppercase{Apurba Das}\authorrefmark{2},
\uppercase{Suman Banerjee}\authorrefmark{3},
\uppercase{Tathagata Ray}\authorrefmark{4}
}

\address[1]{Computer Science and Information Systems, BITS Pilani Hyderabad Campus, Hyderabad, India (e-mail: p20220017@hyderabad.bits-pilani.ac.in)}
\address[2]{Computer Science and Information Systems, BITS Pilani Hyderabad Campus, Hyderabad, India (e-mail: apurba@hyderabad.bits-pilani.ac.in)}
\address[3]{Dept. of Computer Science and Engineering, Indian Institute of Technology Jammu, Jammu \& Kashmir, India (e-mail: suman.banerjee@iitjammu.ac.in)}
\address[4]{Computer Science and Information Systems, BITS Pilani Hyderabad Campus, Hyderabad, India (e-mail: rayt@hyderabad.bits-pilani.ac.in)}

\corresp{Corresponding author: Apurba Das (e-mail: apurba@hyderabad.bits-pilani.ac.in).}

\markboth
{Kiran Mekala \headeretal: Multi-core \& GPU-based Balanced Butterfly Counting}
{Kiran Mekala \headeretal: Multi-core \& GPU-based Balanced Butterfly Counting}

\begin{abstract}

Balanced butterfly counting, corresponding to counting balanced $(2,2)$-bicliques, is a fundamental primitive in the analysis of signed bipartite graphs and provides a basis for studying higher-order structural properties such as clustering coefficients and community structure. Although prior work has proposed an efficient CPU-based serial method for counting balanced $(2,k)$-bicliques. The computational cost of balanced butterfly counting remains a major bottleneck on large-scale graphs. In this work, we present the highly parallel implementations for balanced butterfly counting for both multicore CPUs and GPUs. The proposed multi-core algorithm (M-BBC) employs fine-grained vertex-level parallelism to accelerate wedge-based counting while eliminating the generation of unbalanced substructures. To improve scalability, we develop a GPU-based method (G-BBC) that uses a tile-based parallel approach to effectively leverage shared memory while handling large vertex sets. We then present an improved variation, G-BBC++, which integrates dynamic scheduling to mitigate workload imbalance and maximize throughput. We conduct an experimental assessment of the proposed methods across 15 real-world datasets. Experimental results exhibit that M-BBC achieves speedups of up to $71.13\times$ (average $38.13\times$) over the sequential baseline \bbtwok. The GPU-based algorithms deliver even greater improvements, achieving up to \textbf{13,320$\times$ speedup} (average $2,600\times$) over \bbtwok and outperforming M-BBC by up to $186\times$ (average $50\times$). These results indicate the substantial scalability and efficiency of our parallel algorithms and establish a robust foundation for high-performance signed motif analysis on massive bipartite graphs.
\end{abstract}



\titlepgskip=-15pt

\maketitle

\section{Introduction}
\label{sec:introduction}

    

\PARstart{B}{ipartite} graphs have emerged as a fundamental modeling framework for representing interactions among two distinct types of entities across diverse domains, including user-product relations in e-commerce, collaborations between actors and movies, and connections between authors and reviews in academic peer review systems. Formally, a bipartite graph is represented as $G=(U, V, E)$, in which $U$ and $V$ are disjoint vertex sets; each edge $e= (u,v) \in E$ connects a vertex $u \in U$ with a vertex $v \in V$ (as illustrated in Fig.~\ref{fig:gray_vs_signed}(a)). To ground this in a real-world example, consider an e-commerce bipartite network where customers represent one set of nodes and products represent the other. An edge is established over a customer and a product when the customer has interacted with or acquired the product. The study of structural patterns in bipartite graphs has received considerable attention, with particular focus on cohesive substructures such as bicliques~\cite{yao2022identifying,chen2022efficient,lyu2020maximum}, bi-trusses~\cite{zou2016bitruss,chen2021higher}, bi-cores~\cite{luo2023efficient}, and related dense subgraph models. Among these cohesive structures, the rectangle and the butterfly (a $(2,2)$ complete bipartite subgraph, also called a 4-cycle) are the simplest non-trivial motifs and play a key role in the analysis of more complex subgraph patterns~\cite{wang2014rectangle,aksoy2017measuring}. The butterfly is particularly useful for capturing local connectivity and higher-order interactions; consequently, butterfly counting has emerged as a core computational task in the structural analysis of bipartite networks, supporting various downstream tasks such as measuring cohesion, identifying dense regions, and enabling efficient subgraph enumeration~\cite{wang2019vertex,sanei2018butterfly,wang2022accelerated,sanei2019fleet,shi2020parallel,xu2022efficient}.

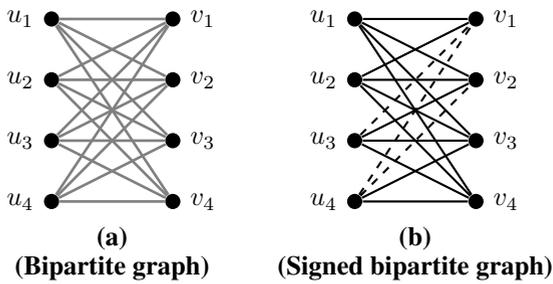
\begin{figure}[t]
    \centering
    \begin{tikzpicture}[scale=0.8]
        \tikzstyle{blackdot} = [circle, fill=black, minimum size=0.2cm, inner sep=0pt]

        \node[blackdot, label=left:$u_1$] (u0) at (0,1.4) {};
        \node[blackdot, label=left:$u_2$] (u1) at (0,0.4) {};
        \node[blackdot, label=left:$u_3$] (u2) at (0,-0.6) {};
        \node[blackdot, label=left:$u_4$] (u3) at (0,-1.6) {};

        \node[blackdot, label=right:$v_1$] (v0) at (2,1.4) {};
        \node[blackdot, label=right:$v_2$] (v1) at (2,0.4) {};
        \node[blackdot, label=right:$v_3$] (v2) at (2,-0.6) {};
        \node[blackdot, label=right:$v_4$] (v3) at (2,-1.6) {};

        \foreach \i in {0,1,2,3} {
            \foreach \j in {0,1,2,3} {
                \draw[line width=1pt, color=gray] (u\i) -- (v\j);
            }
        }

        \node[blackdot, label=left:$u_1$] (bu0) at (5,1.4) {};
        \node[blackdot, label=left:$u_2$] (bu1) at (5,0.4) {};
        \node[blackdot, label=left:$u_3$] (bu2) at (5,-0.6) {};
        \node[blackdot, label=left:$u_4$] (bu3) at (5,-1.6) {};

        \node[blackdot, label=right:$v_1$] (bv0) at (7,1.4) {};
        \node[blackdot, label=right:$v_2$] (bv1) at (7,0.4) {};
        \node[blackdot, label=right:$v_3$] (bv2) at (7,-0.6) {};
        \node[blackdot, label=right:$v_4$] (bv3) at (7,-1.6) {};

        \draw[thick] (bu0) -- (bv0);
        \draw[thick] (bu0) -- (bv1);
        \draw[thick] (bu0) -- (bv2);
        \draw[thick] (bu0) -- (bv3);

        \draw[thick] (bu1) -- (bv0);
        \draw[thick] (bu1) -- (bv1);
        \draw[thick] (bu1) -- (bv2);
        \draw[thick] (bu1) -- (bv3);

        \draw[thick, dashed] (bu2) -- (bv0);
        \draw[thick] (bu2) -- (bv1);
        \draw[thick] (bu2) -- (bv2);
        \draw[thick] (bu2) -- (bv3);

        \draw[thick, dashed] (bu3) -- (bv0);
        \draw[thick, dashed] (bu3) -- (bv1);
        \draw[thick] (bu3) -- (bv2);
        \draw[thick] (bu3) -- (bv3);

        \node at (1, -2.2) {\textbf{(a)}};
        \node at (1, -2.7) {\textbf{(Bipartite graph)}};
        \node at (6, -2.2) {\textbf{(b)}};
         \node at (6, -2.7) {\textbf{(Signed bipartite graph)}};
    \end{tikzpicture}
    \caption{Example of (a) a bipartite graph (gray lines denote edges) and (b) a signed bipartite graph (solid/dashed edges denote positive/negative edges, respectively).}
    \label{fig:gray_vs_signed}
\end{figure}

In many real-world scenarios, the interactions between entities from two different sets are not simply binary existence but also carry sentiment or polarity. While standard bipartite graphs are effective for modeling unsigned pairwise interactions, they are limited in expressiveness, as they assume all edges are unweighted and unsigned. In contrast, real-world relationships often involve both positive and negative feedback, as well as preferences and trust or distrust of information, which cannot be captured by an unsigned model alone. For instance, in an e-commerce platform, a customer may express either a favourable or an unfavourable opinion about a product. To account for richer interaction semantics, the bipartite model is extended to incorporate signed edges, yielding the notion of a signed bipartite graph (as illustrated in Fig.~\ref{fig:gray_vs_signed}(b)). Formally, a signed bipartite graph is defined as $G=(U,V,E=E^+ \cup E^-)$, where $U$ and $V$ are disjoint node sets and $E \subseteq U \times V$ is partitioned into positive edges $E^+$ and negative edges $E^-$. Each edge $e \in E$ carries a sign, with $\mathrm{sign}(e)=+$ for $e \in E^+$ and $\mathrm{sign}(e)=-$ for $e \in E^-$. Neglecting sign information may result in a misleading structural interpretation of the graph, especially when positive and negative interactions carry distinct semantic meanings. Inspired by this need for a more expressive representation, the traditional butterfly counting problem has been broadened to signed bipartite graphs, which eventually led to the notion of a \textbf{balanced butterfly}~\cite{derr2019balance}. A butterfly is known as \textit{balanced} if it contains an even number of negative edges. This signed extension introduces additional computational and structural demands compared to the unsigned case, requiring the development of specific algorithms to efficiently count balanced butterflies in signed bipartite graphs.

Identifying cohesive substructures in signed bipartite graphs is essential for revealing consistent patterns of positive and negative interactions and for facilitating a deeper understanding of network polarity and structural stability. A notable example of this structure is the \textit{balanced signed biclique} (Definition~\ref{def: bsb}), which extends the concept of balanced butterflies to more extensive and denser subgraphs. While most existing work on biclique mining has focused on unsigned bipartite graphs, Chen et al.~\cite{chen2020efficient} were among the first to examine this problem in signed networks. By leveraging principles from balance theory~\cite{heider1946attitudes}, they introduced the notion of a balanced signed biclique, defined as a biclique that includes no unbalanced butterflies, and developed algorithms to enumerate all maximal balanced bicliques that satisfy a given size threshold. Beyond balanced bicliques, several other cohesive patterns have been explored in signed bipartite graphs, including signed bicliques~\cite{sun2023efficient,sun2022maximal,wang2024efficient} and signed bitrusses \cite{chung2023maximum}. However, enumerating maximal balanced bicliques is computationally expensive due to the strict constraints imposed by structural balance. In contrast, balanced butterflies, despite being smaller substructures, are sufficiently expressive to capture local balance properties while remaining computationally simpler. Thus, balanced butterfly counting has emerged as a fundamental building block for analyzing local structure, stability, and polarity in signed networks, serving as an efficient and interpretable unit for capturing local balance, detecting antagonistic interactions, and uncovering stable connectivity patterns in signed bipartite graphs~\cite{karimi2019multi,derr2018congressional,su2009survey}.


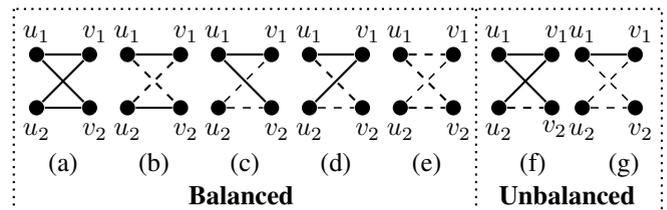
\begin{figure}[ht]
    \centering
    \begin{tikzpicture}[line width=0.2mm]
        \tikzset{
            blackdot/.style={circle, fill=black, minimum size=0.2cm, inner sep=0pt},
        }

         \draw[dotted, thick] (-0.3, 1.3) rectangle (8.25, -1.4);

        \node[blackdot, label={[yshift=1pt]below:$u_2$}] (u0) at (0, 0) {};
        \node[blackdot, label={[yshift=-2pt]above:$u_1$}] (u1) at (0, 0.7) {};
        \node[blackdot, label={[yshift=1pt, xshift=2pt]below:$v_2$}] (v0) at (0.7, 0) {};
        \node[blackdot, label={[yshift=-2pt, xshift=2pt]above:$v_1$}] (v1) at (0.7, 0.7) {};
        \draw[thick] (u0) -- (v0);
        \draw[thick] (u0) -- (v1);
        \draw[thick] (u1) -- (v0);
        \draw[thick] (u1) -- (v1);
        \node at (0.35, -0.75) {(a)};

        \node[blackdot, label={[yshift=1pt]below:$u_2$}] (u2) at (1.2, 0) {};
        \node[blackdot, label={[yshift=-2pt]above:$u_1$}] (u3) at (1.2, 0.7) {};
        \node[blackdot, label={[yshift=1pt, xshift=2pt]below:$v_2$}] (v2) at (1.9, 0) {};
        \node[blackdot, label={[yshift=-2pt, xshift=2pt]above:$v_1$}] (v3) at (1.9, 0.7) {};
        \draw[thick] (u2) -- (v2);
        \draw[dashed, thick] (u2) -- (v3);
        \draw[dashed, thick] (u3) -- (v2);
        \draw[thick] (u3) -- (v3);
        \node at (1.55, -0.75) {(b)};

        \node[blackdot, label={[yshift=1pt]below:$u_2$}] (u4) at (2.4, 0) {};
        \node[blackdot, label={[yshift=-2pt]above:$u_1$}] (u5) at (2.4, 0.7) {};
        \node[blackdot, label={[yshift=1pt, xshift=2pt]below:$v_2$}] (v4) at (3.1, 0) {};
        \node[blackdot, label={[yshift=-2pt, xshift=2pt]above:$v_1$}] (v5) at (3.1, 0.7) {};
        \draw[dashed] (u4) -- (v4);
        \draw[dashed] (u4) -- (v5);
        \draw[thick] (u5) -- (v4);
        \draw[thick] (u5) -- (v5);
        \node at (2.75, -0.75) {(c)};

        \node[blackdot, label={[yshift=1pt]below:$u_2$}] (u6) at (3.6, 0) {};
        \node[blackdot, label={[yshift=-2pt]above:$u_1$}] (u7) at (3.6, 0.7) {};
        \node[blackdot, label={[yshift=1pt, xshift=2pt]below:$v_2$}] (v6) at (4.3, 0) {};
        \node[blackdot, label={[yshift=-2pt, xshift=2pt]above:$v_1$}] (v7) at (4.3, 0.7) {};
        \draw[dashed] (u6) -- (v6);
        \draw[thick] (u6) -- (v7);
        \draw[dashed, thick] (u7) -- (v6);
        \draw[thick] (u7) -- (v7);
        \node at (3.95, -0.75) {(d)};

        \node[blackdot, label={[yshift=1pt]below:$u_2$}] (u8) at (4.8, 0) {};
        \node[blackdot, label={[yshift=-2pt]above:$u_1$}] (u9) at (4.8, 0.7) {};
        \node[blackdot, label={[yshift=1pt, xshift=2pt]below:$v_2$}] (v8) at (5.5, 0) {};
        \node[blackdot, label={[yshift=-2pt, xshift=2pt]above:$v_1$}] (v9) at (5.5, 0.7) {};
        \draw[dashed, thick] (u8) -- (v8);
        \draw[dashed] (u8) -- (v9);
        \draw[dashed, thick] (u9) -- (v8);
        \draw[dashed, thick] (u9) -- (v9);
        \node at (5.15, -0.75) {(e)};

        \node[blackdot, label={[yshift=1pt]below:$u_2$}] (u10) at (6.1, 0) {};
        \node[blackdot, label={[yshift=-2pt]above:$u_1$}] (u11) at (6.1, 0.7) {};
        \node[blackdot, label={[yshift=2pt, xshift=1pt]below:$v_2$}] (v10) at (6.8, 0) {};
        \node[blackdot, label={[yshift=-2pt, xshift=2pt]above:$v_1$}] (v11) at (6.8, 0.7) {};
        \draw[dashed, thick] (u10) -- (v10);
        \draw[thick] (u10) -- (v11);
        \draw[thick] (u11) -- (v10);
        \draw[thick] (u11) -- (v11);
        \node at (6.55, -0.75) {(f)};

        \node[blackdot, label={[yshift=1pt]below:$u_2$}] (u12) at (7.2, 0) {};
        \node[blackdot, label={[yshift=-2pt]above:$u_1$}] (u13) at (7.2, 0.7) {};
        \node[blackdot, label={[yshift=1pt, xshift=2pt]below:$v_2$}] (v12) at (7.9, 0) {};
        \node[blackdot, label={[yshift=-2pt, xshift=2pt]above:$v_1$}] (v13) at (7.9, 0.7) {};
        \draw[dashed] (u12) -- (v12);
        \draw[dashed] (u12) -- (v13);
        \draw[dashed] (u13) -- (v12);
        \draw[thick] (u13) -- (v13);
        \node at (7.75, -0.75) {(g)};

        \node at (2.7, -1.15) {\textbf{Balanced}};
        \node at (7.0, -1.15) {\textbf{Unbalanced}};

\draw[dotted, thick] (5.8, 1.3) -- (5.8, -1.4);

    \end{tikzpicture}
    \caption{An illustration of balanced and unbalanced butterflies.}
    \label{figure:Fig_Two}
\end{figure}

With the growing expansion of signed bipartite graphs in many real-world applications, there is a significant demand for scalable methods that effectively count balanced butterflies. Although numerous parallel techniques have been proposed for butterfly counting on multi-core CPUs~\cite{sanei2018butterfly, shi2022parallel, wang2019vertex, wang2024parallelization}, such methods are specifically developed for unsigned bipartite graphs. In contrast, balanced butterfly counting demands evaluating the sign pattern of each butterfly to determine whether it meets the balance specifications (as shown in Fig.~\ref{figure:Fig_Two} for an example of balanced and unbalanced butterflies). This requirement imposes extra overhead during wedge enumeration and butterfly aggregation. Furthermore, a large fraction of enumerated candidate butterflies are ultimately unbalanced, leading to unnecessary computation and increased runtime. 

In our prior work~\cite{kiran2024efficient}, we tackled this challenge by proposing a bucket-based algorithm that arranges wedges by sign pattern, thereby avoiding unbalanced possibilities early. Although this sequential technique reduces unnecessary work, it does not fully utilise the parallel processing capabilities of modern hardware. The rapid growth in the magnitude and density of signed bipartite graphs requires greater acceleration on multi-core CPUs and GPUs. On the other hand, GPUs offer substantial opportunities to accelerate balanced butterfly counting through their extensive parallelism. Most recently, Xia et al. \cite{xia2024gpu} developed a high-performance GPU algorithm for butterfly counting that effectively addresses issues such as load imbalance, irregular memory access, and synchronisation overhead. However, their method is limited to unsigned bipartite graphs and does not account for the additional sign-specific constraints required for balanced butterfly counting. To the best of our knowledge, 
no prior GPU-based technique simultaneously addresses both high-performance butterfly counting and structural balance enforcement. The proposed work fills this gap by introducing an optimized GPU implementation and a multi-core CPU parallel algorithm, thereby delivering substantial speedups over serial baselines on large signed bipartite graphs.


\subsection{Contributions}
Given the growing prevalence of signed bipartite graphs in real-world applications and the demand for scalable analysis, we develop both multi-core and GPU-based algorithms for balanced butterfly counting. To the best of our knowledge, this is the first study to use both multi-core and GPU architectures for this problem, building a solid basis for scalable signed motif analysis. Our key contributions can be summarized as follows:

\begin{itemize}

\item We employ and expand the BCList++ algorithm, which was originally developed for counting \((p, q)\)-bicliques in unsigned bipartite graphs, to the signed setting. The resulting algorithm, \textbf{SBCList++}, enables efficient counting of balanced \((2, 2)\)-bicliques in signed bipartite graphs.

\item We extend our prior CPU-based algorithm \textbf{\bbtwok} by designing and implementing a vertex-level parallel algorithm, \textbf{M-BBC}, that accelerates wedge-based butterfly counting while avoiding the creation of unbalanced substructures.

\item We introduce \textbf{G-BBC}, a GPU-based algorithm that uses a tile-based, shared-memory processing strategy to efficiently handle large vertex sets.

\item We present \textbf{G-BBC++}, an enhanced GPU implementation that incorporates dynamic workload scheduling to reduce thread-block imbalance and maximize overall performance.

\item We conduct an experimental evaluation of the proposed algorithms on 15 real-world bipartite graphs, demonstrating substantial performance improvements throughout all datasets.

NOTE:
\textit{The implementations of all proposed algorithms will be released as open-source\footnote{\textit{Upon acceptance of the paper}}.}

\end{itemize}

\subsection{Organization of the Paper}

The remainder of this paper has been structured as follows. 
Section~\ref{Sec:related} reviews the existing work on butterfly counting and signed bipartite graphs. 
Section~\ref{Sec:prob_stmt} introduces the preliminaries and formally defines the problem. 
Section~\ref{Sec:alg} presents the proposed algorithms for balanced butterfly counting. 
Section~\ref{sec:load} describes G-BBC++, an enhanced version of G-BBC that incorporates load balancing and dynamic scheduling to improve workload distribution across GPU threads. 
Section~\ref{Sec:EE} discusses the experimental evaluation of the proposed methods. 
Section~\ref{sec: casestudy} presents a detailed case study. Finally,
Section~\ref{Sec:Con} summarises the paper, and Section~\ref{Sec:fut} highlights potential directions for future research.

\section{Related Work}
\label{Sec:related}
       In this section, we review the most relevant earlier studies on structural motif mining in large-scale networks. 
       
       \subsection{Clique Enumeration in Unipartite Graphs}
       Over the past decades, numerous structural patterns have been explored in graphs, including cliques, clans, clubs, cores, and plexes. Among these, we focus on cliques, as they form the foundation for understanding dense subgraph structures and are closely related to our study. A clique consists of a set of vertices such that each pair is connected by an edge. The enumeration of maximal cliques has been extensively studied. The earliest work on this problem was proposed by Akkoyunlu~\cite{akkoyunlu1973enumeration}. This was followed by the seminal work of Bron and Kerbosch~\cite{bron1973algorithm}, who introduced a recursive backtracking algorithm that remains a cornerstone of clique enumeration research. Subsequently studies have extended maximal clique enumeration to various graph settings, including sparse graphs \cite{manoussakis2019new}, graphs generated from spatial data \cite{zhang2019efficient}, temporal graphs \cite{banerjee2019enumeration},large-scale networks \cite{cheng2011finding}, signed unipartite networks \cite{chen2020efficient, li2018efficient,li2019signed,sun2023clique,chen2022balanced}. In addition, numerous algorithms have been designed under different computing paradigms, such as the MapReduce and distributed framework~\cite{lu2010dmaximalcliques}, GPU architectures~\cite{wei2021accelerating}, parallel computing environments~\cite{chen2016parallelizing}, and shared memory systems~\cite{das2018shared}. However, cliques are inherently unipartite structures, and their properties and enumeration techniques do not directly extend to bipartite graphs, which prohibit intra-partition edges and exhibit fundamentally different combinatorial characteristics.

   \subsection{Butterfly Counting in Bipartite Graphs}     
       \par Recently, the study of cliques has been broadened to the bicliques in bipartite graphs. A considerable amount of prior work has investigated bicliques~\cite{chen2022efficient,abidi2020pivot,chen2024maximal,yin2023fairness}. The smallest biclique of practical interest is the $2 \times 2$ biclique, commonly referred to as a butterfly. Numerous sequential methods have been developed for counting butterflies in unsigned bipartite graphs. Wang et al.~\cite{wang2014rectangle} introduced the first vertex-centric butterfly counting algorithm. Subsequently,  Sanei-Mehri et al.~\cite{sanei2018butterfly} proposed a more efficient approach by selecting vertex partitions that minimize the number of wedges. Zhu et al.~\cite{zhu2018fast} further presented an ordering-based algorithm that counts butterflies by processing vertices according to a predefined order. Parallel and distributed variants of butterfly counting algorithms have also been studied extensively. Wang et al.~\cite{wang2014rectangle} developed an MPI-based distributed algorithm that partitions the graph and processes each partition independently. More recently, Shi et al.~\cite{shi2022parallel} proposed a parallel framework incorporating multiple wedge aggregation strategies to improve scalability. Beyond static graphs, butterfly counting has been explored in more complex bipartite settings, including temporal graphs~\cite{cai2023efficient,papadias2024counting}, uncertain graphs~\cite{zhou2021butterfly}, and streaming environments~\cite{meng2024counting}. Additional efforts have focused on optimizing butterfly counting under different computational paradigms, such as the I/O-efficient algorithms~\cite{shi2022parallel}, GPU acceleration~\cite{xia2024gpu}, memory-hierarchy-aware~\cite{wang2024parallelization}, and distributed systems~\cite{tang2024monarch,weng2022distributed}. However, all these approaches are limited to unsigned bipartite graphs and do not consider edge polarity.

   \subsection{Other Subgraph Structures in Signed Bipartite Graphs}  

Derr et al.~\cite{derr2019balance}conducted the first comprehensive study of balance theory in signed bipartite graphs and formally introduced the notion of balanced butterflies as a fundamental unit for assessing structural stability under positive and negative interactions. Sun et al.~\cite{sun2022maximal} investigated the enumeration of maximal balanced signed bicliques, thereby extending balanced motif mining to larger and denser subgraphs in signed bipartite graphs. Chung et al.~\cite{chung2023maximum} proposed the balanced $(k,\epsilon)$-bitruss model, which jointly captures density and tolerance to imbalance, and presented both hardness results and heuristic algorithms for identifying cohesive subgraphs under this model. Although these models capture rich structural properties, they are computationally expensive to enumerate at scale, particularly under strict balance constraints. In our prior work~\cite{kiran2024efficient}, we introduced a bucket-based algorithm that groups wedges by sign patterns, thereby filtering out unbalanced possibilities early. 
      
      \par In recent years, several parallel algorithms have been developed for butterfly counting on multi-core CPUs~\cite{sanei2018butterfly,shi2022parallel,wang2019vertex,wang2024parallelization}. However, these methods are limited to unsigned bipartite graphs. To further improve scalability, GPUs have been widely used to accelerate graph primitives such as triangle counting~\cite{green2014fast}, BFC~\cite{ueno2013parallel}, and PageRank~\cite{wu2010efficient}. Although several GPU-based graph algorithms have been developed, their techniques are not directly applicable to balanced butterfly counting, as the computation patterns and structural requirements differ significantly. We draw inspiration from the recent GPU-based butterfly-counting algorithm, G-BFC, proposed by Xia et al.~\cite{xia2024gpu}. However, extending G-BFC to signed bipartite graphs is non-trivial, as each candidate butterfly must be explicitly checked for balance, introducing significant computational overhead. In the worst case, a signed bipartite graph may contain no balanced butterflies, causing existing unsigned algorithms to incur substantial wasted computation. In this paper, we address this gap by designing both a thread-parallel CPU algorithm and a GPU-based solution tailored for signed bipartite graphs.

\begin{table}[t]
\centering
\caption{Symbols and meanings.}
\label{tab:symbols}
\begin{tabular}{ll}
\hline
\textbf{Symbol} & \textbf{Meaning} \\
\hline
$G = (U, V, E^+, E^-)$ & A signed bipartite graph with signs \\
$|U \cup V|$ & Total number of vertices in graph  $G$ \\
$E(G)$ & Set of all edges in graph  $G$ (i.e., $E^+ \cup E^-$) \\
$E^+$ & Set of positive edges \\
$E^-$ & Set of negative edges \\
$u, v$ & Vertices in $G$ \\
$\lor(u_i, v_j, u_k)$ & A wedge (2-path) in graph $G$ \\

$\lor^s(u_i, v_j, u_k)$ & A Symmetric wedge (2-path) in  graph $G$ \\

$\lor^a(u_i, v_j, u_k)$ & An asymmetric wedge (2-path) in  graph $G$ \\

$(u_i, u_j, v_i, v_j)$ & A butterfly in $G$ \\
$\Gamma(u)$ & Neighbor set of vertex $u$ \\
$N^2_G(u)$ & 2-hop neighbors of vertex $u$ in $G$ \\
$\deg(u)$ & Degree of vertex $u$ \\
$p(u)$ & Priority of vertex $u$\\

$\balancebutterfly$ & balanced butterfly \\
\hline
\end{tabular}
\end{table}
  
\section{Preliminaries and Problem Definition}
\label{Sec:prob_stmt}


We consider a signed bipartite graph $G = (U, V, E = E^{+}\cup E^{-})$, in which $U$ and $V$ are the bipartitions, and $E \subseteq U \times V$ is the edge set partitioned into positive edges $E^+$ and negative edges $E^-$. We define $\text{sign}(e)= ``+"$ for an edge $e \in E^{+}$ and $\text{sign}(e)= ``-"$ for an edge $e \in E^{-}$. For a vertex $u\in U$, let $d(u)$ refer to the degree of $u$ and $\Gamma(u)$ indicate the neighbors of $u$ in $G$. Now, we define some basic terms for our problem.


\begin{definition}[\textbf{Butterfly}~\cite{wang2019vertex}]
Given a bipartite graph $G=(U,V,E)$, a butterfly is a cycle of length four induced by vertices $(u_i, u_j, v_i, v_j)$, where $u_i, u_j \in U$ and $v_i, v_j \in V$, such that all four possible edges between $\{u_i,u_j\}$ and $\{v_i,v_j\}$ exist in $G$.  
In a signed bipartite graph, a \textbf{signed butterfly} is a butterfly whose edges are assigned signs.
\label{def:butterfly}
\end{definition}


\begin{definition}[\textbf{Balanced Butterfly}~\cite{derr2019balance}]
Given a signed bipartite graph $G$, a signed butterfly $b$ in $G$ is considered \textbf{balanced} if it contains an even number of negative edges.
\label{def:bal_butterfly}
\end{definition}

\begin{definition}[\textbf{Biclique}~\cite{derr2019balance}]
Given a bipartite graph $G = (U, V, E)$, a biclique is a complete bipartite subgraph of $G$. A \textbf{signed biclique} is a biclique in which the edges are positive or negative.
\end{definition}


\begin{definition}[\textbf{Balanced Signed biclique}]
Given a signed bipartite graph $G = (U,V,E)$, a balanced signed biclique is an induced subgraph $B$ of $G$ that satisfies:
\begin{itemize}
    \item \textbf{Cohesiveness constraint:}~$B$ is a signed biclique of $G$;
    \item \textbf{Balance constraint:}~$B$ does not contain any unbalanced butterfly.
\end{itemize}
\label{def: bsb}
\end{definition}

\begin{definition}[\textbf{Balanced $(2,k)$-Biclique}]~A $(2,k)$-biclique is balanced if there are no unbalanced butterflies in it.
\label{bal-2-3-bcl}
\end{definition}


\begin{definition}[\textbf{Vertex Priority}~\cite{wang2019vertex}]
Let $a,b \in \{U \cup V\}$ be two vertices. We say that vertex $a$ has higher priority than vertex $b$, denoted by $p(a) > p(b)$, if one of the following conditions holds:
\begin{enumerate}
    \item $degree(a) > degree(b)$;
   \item $id(a)>id(b)$, if $degree(a)=degree(b)$
\end{enumerate}
Here, $id(a)$ is the vertex ID of $a$.
\label{def-priority}
\end{definition}




\begin{definition}[\textbf{Wedge ($\lor$)}~\cite{wang2019vertex}]
Given a bipartite graph $G=(U,V,E)$, let $u_i,u_k \in U$ and $v_j \in V$. 
A wedge $\lor(u_i,v_j,u_k)$ is a path of length two that starts at $u_i$, passes through $v_j$, and ends at $u_k$, as illustrated in Fig.~\ref{fig:wedge}. 
We assume an ordering on vertices in $U$ such that $p(u_i) > p(u_k)$.
\label{def:wedge}
\end{definition}


\begin{figure}[ht]
\centering
\begin{tikzpicture}
    \tikzstyle{blackdot} = [circle, fill=black, minimum size=2.5mm, inner sep=0pt]

    \node[blackdot, label=left:{$u_i$}] (u1a) at (0, 0.7) {};
    \node[blackdot, label=left:{$u_k$}] (w1a) at (0, -0.7) {};
    \node[blackdot, label=right:{$v_j$}] (v1a) at (1.1, 0) {};

    \draw[gray, line width=0.6pt] (u1a) -- (v1a);
    \draw[gray, line width=0.6pt] (w1a) -- (v1a);

\end{tikzpicture}
\caption{A wedge in an unsigned bipartite graph.}
\label{fig:wedge}
\end{figure}
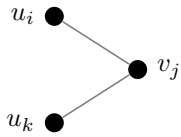


\begin{definition}[\textbf{Symmetric Wedge ($\lor^s$)}~\cite{kiran2024efficient}]
A wedge $\lor(u_i,v_j,u_k)$ in a signed bipartite graph $G=(U,V,E)$ is called a \textbf{symmetric wedge} if the two edges $(u_i,v_j)$ and $(u_k,v_j)$ have the same sign, i.e., both are positive or both are negative. Such a wedge is denoted by $\lor^s(u_i,v_j,u_k)$ and is illustrated in Fig.~\ref{fig:Swedge}.
\label{def:wedgesym}
\end{definition}

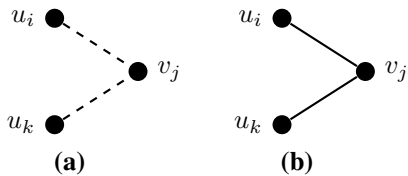
\begin{figure}[ht]
\centering
\begin{tikzpicture}
    \tikzstyle{blackdot} = [circle, fill=black, minimum size=2.5mm, inner sep=0pt]

    \node[blackdot, label=left:{$u_i$}] (u1b) at (3.0, 0.7) {};
    \node[blackdot, label=left:{$u_k$}] (w1b) at (3.0, -0.7) {};
    \node[blackdot, label=right:{$v_j$}] (v1b) at (4.1, 0) {};

    \draw[black, dashed, line width=0.8pt] (u1b) -- (v1b);
    \draw[black, dashed, line width=0.8pt] (w1b) -- (v1b);
     \node at (3.2, -1.2) {\textbf{(a)}};

    \node[blackdot, label=left:{$u_i$}] (u1c) at (6.0, 0.7) {};
    \node[blackdot, label=left:{$u_k$}] (w1c) at (6.0, -0.7) {};
    \node[blackdot, label=right:{$v_j$}] (v1c) at (7.1, 0) {};

    \draw[black, line width=0.8pt] (u1c) -- (v1c);
    \draw[black,  line width=0.8pt] (w1c) -- (v1c);
     \node at (6.2, -1.2) {\textbf{(b)}};

\end{tikzpicture}
\caption{Symmetric wedges (a) and (b) in a signed bipartite graph.}
\label{fig:Swedge}
\end{figure}


\begin{definition}[\textbf{Asymmetric Wedge ($\lor^a$)}~\cite{kiran2024efficient}]
A wedge $\lor(u_i,v_j,u_k)$ in a signed bipartite graph $G=(U,V,E)$ is called an \textbf{asymmetric wedge} if the two edges $(u_i,v_j)$ and $(u_k,v_j)$ have different signs, i.e., one is positive and the other is negative. Such a wedge is denoted by $\lor^a(u_i,v_j,u_k)$ and is illustrated in Fig.~\ref{fig:Awedge}.
\label{def:wedgeasym}
\end{definition}

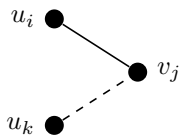
\begin{figure}[ht]
\centering
\begin{tikzpicture}
    \tikzstyle{blackdot} = [circle, fill=black, minimum size=2.5mm, inner sep=0pt]

    \node[blackdot, label=left:{$u_i$}] (u1a) at (0, 0.7) {};
    \node[blackdot, label=left:{$u_k$}] (w1a) at (0, -0.7) {};
    \node[blackdot, label=right:{$v_j$}] (v1a) at (1.1, 0) {};

    \draw[black, line width=0.6pt] (u1a) -- (v1a);
    \draw[black,dashed, line width=0.6pt] (w1a) -- (v1a);

\end{tikzpicture}
\caption{Asymmetric wedge in a signed bipartite graph.}
\label{fig:Awedge}
\end{figure}


\noindent\textbf{Problem Statement.}~
We consider a signed bipartite graph $G=(U,V,E)$. The aim is to count the total number of balanced butterflies in $G$.

\section{proposed algorithms}
\label{Sec:alg}

\begin{algorithm}[t]
\small
\DontPrintSemicolon
	\SetKwInOut{Input}{Input}
	\SetKwInOut{Output}{Output}
	\Input{$G$$(U,V,E)$: a signed bipartite graph.}
	\Output{$b_{2k}$: total number of balanced  $(2,k)$ bicliques.}
        Calculate $p(u)$ for each $u \in \{U \cup V\}$  [\textbf{Definition}~\ref{def-priority}]\\

        For each vertex $u$, sort $N(u)$ according to their priorities.

        $b_{2k}\gets 0$\;

       \ForEach{$u\in U$}{
       $B_1\gets 0$\;
       $B_2\gets 0$\;
        \ForEach{$v \in \Gamma(u)$}{
        \ForEach{$w \in \Gamma(v)$ $|$ $p(w) > p(u)$}{
        \If{$\lor^{s}(u,v,w)$}{
        $B_1[w]++$\;
        }
        \Else{
        $B_2[w]++$\;
        }
        }
        }
        \ForEach{$w\in B_1$}{

            $b_{2k}\gets b_{2k} + \binom{B_1[w]}{k}$\;
 
        }
        \ForEach{$w\in B_2$}{

            $b_{2k}\gets b_{2k} + \binom{B_2[w]}{k}$\;
        }

        }
 return $b_{2k}$
\caption{\bbtwok (where $k = 2$)\cite{kiran2024efficient}}
\label{Algorithm:Bucket}
\end{algorithm}
In this section, we first explain the baseline and reference algorithm relevant to our proposed method.

A straightforward baseline for counting balanced butterflies follows a two-step procedure: first, enumerate all butterflies in the graph, and then verify which of them satisfy the balance condition. To implement this approach, we leverage the 
$(p,q)$-biclique enumeration algorithm \textbf{BCList++}~\cite{yang2023p} to enumerate all $(2,2)$-bicliques, also known as butterflies in the signed bipartite graph. We then extended BCList++ by incorporating an additional balance check for each enumerated butterfly. We refer to this extended baseline as \textbf{SBCList++}, which serves as a reference method for counting balanced $(2,2)$-bicliques in signed bipartite graphs.
Next, we introduce a second reference algorithm,-\textbf{\bbtwok}, upon which our proposed algorithms are built. This method is designed to count balanced $(2,k)$-bicliques in signed bipartite graphs, and its pseudocode is presented in Algorithm~\ref{Algorithm:Bucket}.

\begin{lemma}
Any balanced butterfly in a signed bipartite graph can be formed either using a pair of symmetric wedges or a pair of asymmetric wedges, but not both.
\label{lemma:lem-1}
\end{lemma}
\begin{proof}
Considering a signed bipartite graph $ G = (U, V, E)$, let $butterfly({u,w,v,x})$ represent a balanced butterfly in $G$ in which $u,w \in U$ and $v,x \in V$ with two wedges $\lor(u,v,w)$ and $\lor(u,x,w)$.

\noindent\textbf{case-1.}~$\lor^{s}(u,v,w)$ and $\lor^{s}(u,x,w)$. In this case, the number of negative edges is $0$, $2$, or $4$, thus even.

\noindent\textbf{case-2.}~$\lor^{a}(u,v,w)$ and $\lor^{a}(u,x,w)$. In this case, the number of negative edges is $2$, thus even.

\noindent\textbf{case-3.}~$\lor^{s}(u,v,w)$ and $\lor^{a}(u,x,w)$. In this case, the number of negative edges is $1$ or $3$, thus odd. So, this combination cannot make a balanced butterfly.

This completes the proof.
\end{proof}



\subsection{\bbtwok Algorithm} 
A core idea of Algorithm~\ref{Algorithm:Bucket} is to systematically organize symmetric and asymmetric wedges so that balanced $(2,k)$-bicliques can be counted by evaluating the number of ways to select $k$ vertices from the centers of grouped wedges. To facilitate this, we use two hash maps, referred to as buckets, denoted by $B_1$ and $B_2$, where each key corresponds to a center vertex $w$, and the associated value represents the number of wedges in which $w$ participates.

The algorithm proceeds by iterating on each vertex $u \in U$. For every two-hop neighbor $w$ reachable from $u$, we count how many symmetric wedges involve $w$ and store this value in $B_1[w]$. Similarly, the number of asymmetric wedges involving $w$ is recorded in $B_2[w]$. Finally, to count balanced $(2,k)$-bicliques, we compute the number of possible ways to choose $k$ wedges from $B_1[w]$ (in the case of symmetric wedges) or from $B_2[w]$ (in the case of asymmetric wedges), depending on the type of balance being considered. Our algorithm is agnostic to the specific value of $k$; once wedges are grouped, we compute balanced butterflies using $\binom{n}{k}$-style aggregation (where $k = 2$). 

The \bbtwok algorithm described above provides an efficient \emph{sequential} approach for counting balanced $(2,k)$-bicliques, in which the size-two side $U$ is fixed by definition. In this work, however, we focus on the design of parallel algorithms for \emph{balanced butterfly counting}, corresponding to the $(2,2)$ case.

Unlike the general $(2,k)$ setting, balanced butterfly counting is symmetric with respect to the two partitions of the bipartite graph. This symmetry allows the counting process to be anchored on either partition without affecting correctness. Our multi-core and GPU algorithms utilise this property by dynamically choosing the partition with fewer vertices as the processing side, leading to reducing the number of candidate wedges and enhancing workload balance on parallel architectures.

We next present M-BBC, a multi-core algorithm, and then introduce G-BBC and G-BBC++ algorithms, designed for counting balanced butterflies in signed bipartite graphs.





\subsection{Multi-core parallelization}

The transition from serial to parallel implementation of the proposed algorithm posed several key challenges, including race conditions in shared data structures, redundant butterfly counting, and workload imbalance across threads.

To address these, we developed the M-BBC with careful attention to thread safety and scalability. We employ thread-local \emph{ConcurrentMap} to store wedge counts per vertex $u$, guaranteeing that no other thread accesses these maps when a thread is building wedge counts, accordingly eliminating the race conditions. A strict vertex-priority ordering ensures that each butterfly is counted exactly once. We further adopt dynamic scheduling to improve load balancing across threads, using TBB's work-stealing mechanism to automatically distribute tasks to available worker threads. In this, the computation of balanced butterflies is decomposed into fine-grained tasks corresponding to vertices and their associated neighbor processing. Each thread maintains a local queue of tasks; when a thread completes its assigned work, it steals pending work from another thread's queue, ensuring that a thread finishing early remains productive by assisting with tasks assigned to busier threads, effectively mitigating the impact of skewed or high-degree vertices.

\begin{algorithm}[t!]
\small
	\SetKwInOut{Input}{Input}
	\SetKwInOut{Output}{Output}
	\Input{$G$$(U,V,E)$, A signed bipartite graph.}
	\Output{Total number of $\balancebutterfly$.}
	\SetKwFunction{FMain}{BBFC}
    \SetKwProg{Fn}{Function}{}{}
    \Fn{\FMain{$G$}}{

$\balancebutterfly \gets 0$\\
$S \gets V$ \\
\If{$|V| > |U|$}{$S \gets U$}

       \ForPar{$u \in \{S\}$}{
       $B_1\gets 0$\ \\
       $B_2\gets 0$\ \\ 
        \ForPar{$v \in \Gamma(u)$ }{ 
        \ForPar{$w \in \Gamma(v)$ $|$ $p(w) < p(u)$}{
        \If{$\sigma(u,v) == \sigma(v,w)$}{
        $B_1[w]++$\
        }
        \Else{
        $B_2[w]++$\
        }
        }
        }
        \ForPar{$w\in B_1$}{
            $\balancebutterfly \gets\balancebutterfly + \binom{B_1[w]}{2}$    
        }
        \ForPar{$w\in B_2$}{
            $\balancebutterfly \gets\balancebutterfly + \binom{B_2[w]}{2}$    
        }

        }
  }
  return $\balancebutterfly$
\caption{M-BBC: Multi-core balanced butterfly counting}
\label{Algorithm:parallel}
\end{algorithm}

Next, nested parallelism is adopted for processing neighbors and $2$-hop neighbors, allowing tasks at multiple levels to run concurrently and be dynamically stolen by idle threads. By combining thread-local data structures for intermediate wedge counts with atomic operations strictly for the global butterfly accumulation, the implementation achieves high concurrency. This dynamic scheduling strategy is more effective for graphs with diverse degree distributions, preventing idle threads caused by uneven task distribution and ensuring near-optimal utilization of all available cores.

Finally, in contrast to the serial version, the multi-core implementation allocates wedge-related data structures locally to each thread and performs thread-safe aggregation at the outermost level, considerably reducing memory contention while maintaining correctness under parallel execution.

\subsubsection{M-BBC Algorithm}

The M-BBC algorithm counts the number of balanced butterflies in a signed bipartite graph $G = (U, V, E)$ by the advantage of vertex-level parallelism. To reduce redundant computation, the algorithm initially chooses the smaller side, $S = \min(U, V)$, as the processing side. This choice limits the set of vertices for which wedge enumeration is performed, thereby improving efficiency. The complete pseudocode is detailed in Algorithm~\ref{Algorithm:parallel}. 

One of the main challenges in parallelization arises when concurrent updates to wedge counts are directed to the appropriate buckets. To address this, we employ a \textit{parallel hash table} that supports insertion, deletion, and membership queries. The algorithm iterates over each vertex $u \in S$ in parallel. For an anchor vertex $u$, two additional data structures, $B_1$ and $B_2$, are initialized to maintain counts of symmetric and asymmetric wedges, respectively, specified by their endpoint vertex $w$. For each anchor vertex $u$, the algorithm checks all neighbors $v \in \Gamma(u)$ and, for each such neighbor, it traverses 2-hop neighbors $w \in \Gamma(v)$ that meet the priority condition $p(w) < p(u)$. This condition ensures that each butterfly is counted exactly once, therefore avoiding duplicate enumeration.

For every valid wedge $(u,v,w)$, the algorithm compares the signs of edges $(u,v)$ and $(v,w)$. If the two signs are equal, the wedge is classified as symmetric, and the corresponding counter in $B_1[w]$ is incremented; otherwise, $B_2[w]$ is incremented. Each such update is performed in constant time. Once all wedges associated with the vertex $u$ have been processed, the algorithm aggregates the counts stored in $B_1$ and $B_2$. For each endpoint vertex $w$, the number of balanced butterflies contributed by symmetric wedges is calculated as $\binom{B_1[w]}{2}$, and similarly for asymmetric wedges using $\binom{B_2[w]}{2}$. These values are incorporated into the global balanced butterfly counter. Since the vertices in $S$ are processed independently, the algorithm naturally supports parallel execution. The final output is the total number of balanced butterflies in the graph, stored in $\balancebutterfly$.



\subsubsection{Complexity Analysis} 

We analyze the time complexity of Algorithm~\ref{Algorithm:parallel}. The algorithm processes vertices from the smaller partition $S = \min(U,V)$ to reduce redundant computation. For each vertex $u \in S$, it explores all two-hop paths (or wedges) of the form $(u,v,w)$, where $v \in \Gamma(u)$ and $w \in \Gamma(v)$, subject to the priority condition $p(w) < p(u)$. This condition ensures that each butterfly is counted exactly once.

For a fixed vertex $u$, the running time is proportional to the number of such wedges, which is given by $\sum_{v \in \Gamma(u)} |\Gamma(v)|$. Each wedge is processed in constant time by updating the auxiliary structures $B_1$ and $B_2$. The final aggregation step iterates over these structures and is therefore linear in their sizes, which are bounded by the number of processed wedges.
By summing over all vertices in the processing side $S$, the time complexity is  $O(\sum_{u \in S} \sum_{v \in \Gamma(u)} |\Gamma(v)|)$. The outer loop vertices in $S$ are parallel since each vertex can be processed independently. Assuming that the workload is evenly distributed across $P$ threads, the total computation is shared uniformly, and the parallel running time becomes $O(\frac{1}{P}\sum_{u \in S} \sum_{v \in \Gamma(u)} |\Gamma(v)|)$.


\subsection{GPU beased parallelization}

Although serial computation of balanced butterflies in signed bipartite graphs has been studied, no existing work has explored GPU-based parallel algorithms for counting balanced butterflies. The main challenge lies in the massive combinatorial complexity of balanced butterfly enumeration, which involves traversing wedges with a specific sign pattern. Traditional CPU-based approaches, even when parallelized, rely on thread-local memory and dynamic scheduling to avoid race conditions. However, such strategies cannot be directly applied to GPUs, which require thousands of lightweight threads to exploit massive parallelism and have limited per-thread memory due to GPU hardware constraints. For example, allocating even a modest $0.1$ MB buffer per thread would exceed the GPU's memory capacity for millions of threads. Therefore, designing a fine-grained GPU-parallel version of balanced butterfly counting is essential to fully utilize the computational power of modern GPUs while ensuring correctness and efficiency. Similar to GPU-based butterfly counting \cite{xia2024gpu}, we face several challenges, including load imbalance, memory hierarchy optimization, and synchronization. In addition, our problem requires handling the balanced sign-based condition, which introduces new computational and memory challenges not considered in previous studies. 

\begin{figure}[ht]
\centering
\begin{tikzpicture}[node distance=0.55cm and 0.45cm]
\tikzstyle{whitedot} = [circle, draw=black, fill=black, minimum size=1.8mm, inner sep=0pt]
\tikzstyle{blackdot} = [circle, fill=black, minimum size=1.8mm, inner sep=0pt]

\node[anchor=west] at (0,3.7) {$\mathbf{Block_0}$};

\node at (0,2.8) {\footnotesize U(G)};
\node[whitedot, label=above:{\footnotesize$u_0$}] (U0) at (0.5,2.8) {};
\node at ($(U0)+(0.5,0)$) {\dots};

\node[anchor=west] at (2.5,2.2) {\footnotesize \textit{Threads parallel}};
\node[anchor=west] at (2.5,0.9) {\footnotesize \textit{Thread$_0$ updates shared map by tiles}};

\node at (0,1.5) {\footnotesize V(G)};
\node[blackdot, label=above:{\footnotesize$v_0$}] (V0) at (0.5,1.5) {};
\node[blackdot, label=above:{\footnotesize$v_1$}] (V1) at ($(V0)+(0.68,0)$) {};
\node at ($(V1)+(0.5,0)$) {\dots};
\node[blackdot, label=above:\footnotesize{$v_i$}] (Vi) at ($(V1)+(1.0,0)$) {};

\draw[-, thick] (U0) -- ($(V0.north) + (0, 3.5mm)$);
\draw[-, dashed] (U0) -- ($(V1.north) + (0, 3.5mm)$);
\draw[-, thick] (U0) -- ($(Vi.north) + (0, 3.5mm)$);

\node at (0,0.2) {\footnotesize U(G)};
\node[whitedot, label=above:{\footnotesize$u_0$}] (U0b) at (0.5,0.2) {};
\node[whitedot, label=above:{\footnotesize$u_1$}] (U1b) at ($(U0b)+(0.68,0)$) {};
\node at ($(U1b)+(0.5,0)$) {\dots};
\node[whitedot, label=above:{\footnotesize$u_{128}$}] (U128b) at ($(U1b)+(1.0,0)$) {};
\node[whitedot, label=above:{\footnotesize$u_{129}$}] (U129b) at ($(U128b)+(0.68,0)$) {};
\node at ($(U129b)+(0.5,0)$) {\dots};
\node[whitedot, label=above:{\footnotesize$u_{256}$}] (U256b) at ($(U129b)+(1.0,0)$) {};
\node[whitedot, label=above:{\footnotesize$u_{257}$}] (U257b) at ($(U256b)+(0.68,0)$) {};
\node at ($(U257b)+(0.5,0)$) {\dots};
\node at ($(U257b)+(0.5,0)$) {\dots};
\node[whitedot, label=above:{\footnotesize$u_{n-1}$}] (Un1b) at ($(U257b)+(1.0,0)$) {};
\node[whitedot, label=above:{\footnotesize$u_{n}$}] (Unb) at ($(Un1b)+(0.68,0)$) {};

\draw[-, thick] (V0) -- ($(U0b.north) + (0, 3.5mm)$);
\draw[-, thick] (V0) -- ($(U1b.north) + (0, 3.5mm)$);
\draw[-, dashed] (V0) -- ($(U128b.north) + (0, 3.5mm)$);

\node (endLoop1) at ($(U0b)+(1.0,0)$) {};   
\node (endLoop2) at ($(U128b)+(1.0,0)$) {}; 
\node (endLoop3) at ($(U256b)+(1.0,0)$) {}; 
\node (endLoop4) at ($(Un1b)+(1.0,0)$) {}; 
\node[
    draw,
    thick,
    rounded corners,
    fit=(U0b) (endLoop1),
    inner sep=3pt,
    xshift=3pt,
    yshift= -1pt,
    fill=gray,
    fill opacity=0.2
] (box1) {};
\node[below=-1pt of box1, anchor=north] {\footnotesize $loop_1$};

\node[
    draw,
    thick,
    rounded corners,
    fit=(U128b) (endLoop2),
    inner sep=3pt,
    xshift=3pt, 
    yshift= -1pt,
    fill=gray,
    fill opacity=0.2
] (box2) {};
\node[below=-1pt of box2, anchor=north] {\footnotesize $loop_2$};
\node[
    draw,
    thick,
    rounded corners,
    fit=(U256b) (endLoop3),
    inner sep=3pt,
    xshift=3pt,
    yshift= -1pt,
    fill=gray,
    fill opacity=0.2
] (box3) {};
\node[below=-1pt of box3, anchor=north] {\footnotesize $loop_3$};

\node[
    draw,
    thick,
    rounded corners,
    fit=(Un1b) (endLoop4),
    inner sep=3pt,
    xshift=3pt,
    yshift= -1pt,
    fill=gray,
    fill opacity=0.2
](box4) {};
\node[
    below=-1pt of box4,
    anchor=north
] {\footnotesize $ \mathrm{loop}_{\frac{n}{\mathrm{sharesize}}} $};

\node[align=center] at (0,-0.4) {\shortstack{\footnotesize shared \\[-0.2em] \footnotesize map}};
\end{tikzpicture}
\caption{Usage of tiles in shared map. }
\label{fig:chunk}
\end{figure}
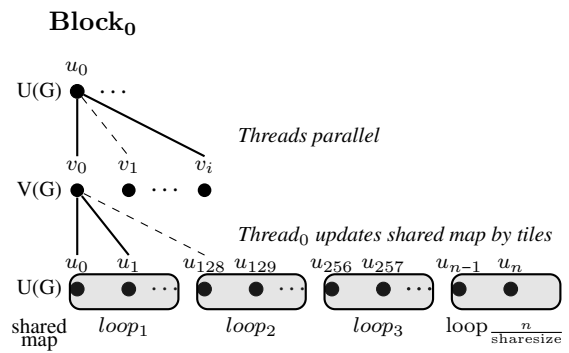

\subsubsection{G-BBC Algorithm}

We present the G-BBC algorithm in Algorithm~\ref{algo:gpu-bfc}, which leverages GPU parallelism to efficiently count balanced butterflies in signed bipartite graphs. GPU blocks process vertices $u$ from the smaller partition in a grid-stride manner, and the threads within that block cooperatively traverse the neighbors of $u$. To support sign-aware two-hop aggregation, each block allocates two shared-memory bitmaps that store intermediate counts of the symmetric and asymmetric wedge types, corresponding to the sign patterns $\lor^s$ and $\lor^a$ (Definitions~\ref{def:wedgesym} and \ref{def:wedgeasym}). Because shared memory cannot store counters for the entire two-hop neighborhood $N_G^{2}(u)$ at once, the algorithm partitions $N_G^{2}(u)$ into contiguous tiles (as illustrated in Fig.~\ref{fig:chunk}). In each iteration, a block loads one tile of size $T$ into shared memory, initializes the corresponding bitmap entries, and then has its threads accumulate sign-consistent wedge counts by exploring all two-hop paths $(u, v, w)$ whose endpoint $w$ lies within the current tile. After processing a tile, the block reduces the bitmap counters using warp-level and block-level shuffle operations, yielding the balanced butterfly contribution for that tile. This load–count–reduce tiled workflow repeats until all tiles covering $N_G^{2}(u)$ have been processed, ensuring efficient shared-memory usage and maintaining high GPU parallel throughput.

\begin{algorithm}[t!]
\small
\SetKwInOut{Input}{Input}
\SetKwInOut{Output}{Output}
\Input{$G(U, V, E, \sigma)$: A signed bipartite graph.}
\Output{$B(G)$: Total number of  $\balancebutterfly.$\\}

 $\balancebutterfly \gets 0$\\
$S \gets V$ \\
\If{$|V| >|U|$}{$S \gets U$}
\ForEach{$u = \text{blockIdx.x}$ and $u.id < S $}{
\If{$threadIdx.x = 0$}{initialize two shared hashmaps $B_1$, $B_2$ with $0$ }

    \ForEach{$v \in \Gamma(u)$}{
        \ForEach{$w \in \Gamma(v)$}{
            \If{$w.\text{id} > u.\text{id}$}{
              \If{$\sigma(u, v) == \sigma(v, w)$}{
    $atomicAdd$($\texttt{B}_1[w]$)

}
\Else{
    $atomicAdd$($\texttt{B}_2[w])$
}
            }
        }
    }
    $syncthreads();$

    $\balancebutterfly \gets$ $blockreduce$($ \binom{B_1[w]}{2} + \binom{B_2[w]}{2}$) \\
    
$u.id += gridDim$
   
}
return $\balancebutterfly$

\caption{G-BBC: GPU-based balanced butterfly counting.}
\label{algo:gpu-bfc}
\end{algorithm}

\textbf{G-BBC Algorithm:}
We process the flow of each vertex as follows. First, we select the partition with fewer vertices as the starting layer, which is denoted as S (Lines 2-4). We initialize the two shared bitmaps ($B_1$ and $B_2$) to $0$ (lines 6-7) by passing this task to the first thread of each block. Second, threads cooperatively traverse the adjacency list of the vertex and explore its neighbors (Line 8), and access the corresponding two-hop neighbors (Line 9). Third, we judge whether the two-hop neighbors are stored in a symmetric or an asymmetric bitmap based on the signed patterns. If so, we safely update the shared bitmap using atomic operations (lines 11-14). Here, we use $atomicAdd(B_1[w])$ and $atomicAdd(B_2[w])$ to update the shared map, as multiple threads can access it simultaneously. Then we need to do $syncthreads()$, which acts as a synchronization barrier to guarantee the consistency of the shared bitmaps across threads (Line 15). Once all neighbors are processed, each thread performs a local aggregation over its portion of the bitmap to compute partial counts of balanced butterflies. A warp-level reduction using \texttt{\_\_shfl\_down\_sync} is applied to sum the counts within each warp. The sums per warp are stored in a shared array \texttt{warp\_sums}, followed by a final reduction throughout the block performed by $thread_0$. The result is atomically accumulated into the global sum \texttt{sums}. Finally, the block proceeds to process another vertex after it finishes processing the current workload (Line 17).

\subsubsection{Complexity Analysis}

In Algorithm~\ref{algo:gpu-bfc}, the initialization steps, including selecting the smaller partition $S = \min(U, V)$ and allocating shared memory data structures, incur constant overhead per block. The algorithm processes vertices $u \in S$ using GPU blocks in a grid-stride manner. However, the asymptotic analysis is based on total work rather than assuming ideal block-level concurrency. For a fixed vertex $u$, the algorithm enumerates all two-hop paths (wedges) of the form $(u,v,w)$ where $v \in \Gamma(u)$ and $w \in \Gamma(v)$, subject to an orientation constraint that avoids duplicate counting. Each valid wedge results in a constant-time atomic update to the shared-memory bitmaps. Consequently, the total amount of work performed for vertex $u$ is proportional to $\sum_{v \in \Gamma(u)} \deg(v)$. After wedge enumeration, the algorithm performs block-level aggregation using shared-memory and warp-level reductions. This step is linear in the number of distinct wedge endpoints and bounded by the number of processed wedges, with constant synchronization overhead per block. Summing over all vertices in the processing side $S$, the total work complexity of the G-BBC algorithm is $ O\left(\sum_{u \in S} \sum_{v \in \Gamma(u)} \deg(v) \right) $. Although GPU parallelism significantly reduces wall-clock runtime in practice, it does not alter the asymptotic work complexity of the algorithm.


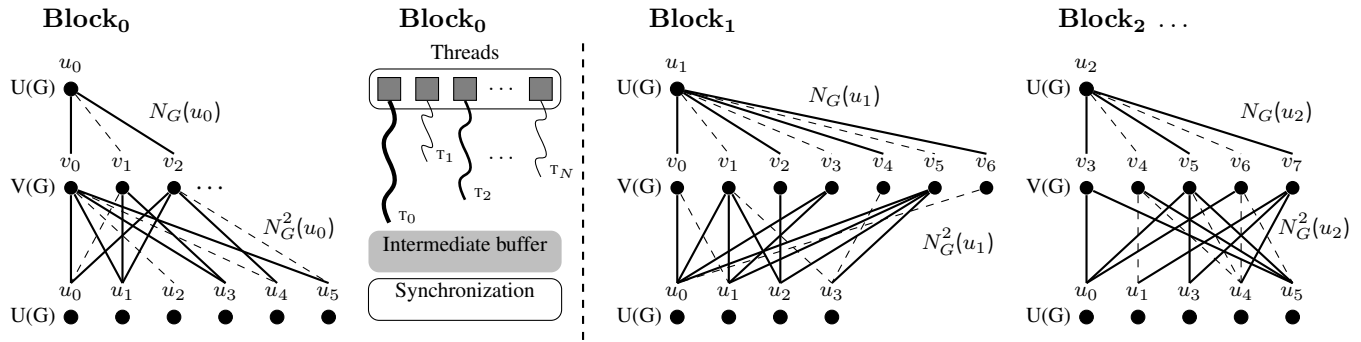
\begin{figure*}[ht]
\begin{tikzpicture}[node distance=0.55cm and 0.45cm]
\tikzstyle{whitedot} = [circle, draw=black, fill=black, minimum size=1.8mm, inner sep=0pt]
\tikzstyle{blackdot} = [circle, fill=black, minimum size=1.8mm, inner sep=0pt]

\node[anchor=west] at (0,3.7) {$\mathbf{Block_0}$};

\node at (0,2.8) {\footnotesize U(G)};
\node[whitedot, label=above:{\footnotesize$u_0$}] (U0) at (0.5,2.8) {};

\node at (2,2.5) {\footnotesize $N_G(u_0)$};
\node at (3.5,1.0) {\footnotesize $N_G^2(u_0)$};

\node at (0,1.5) {\footnotesize V(G)};
\node[blackdot, label=above:{\footnotesize$v_0$}] (V0) at (0.5,1.5) {};
\node[blackdot, label=above:{\footnotesize$v_1$}] (V1) at ($(V0)+(0.68,0)$) {};
\node[blackdot, label=above:\footnotesize{$v_2$}] (V2) at ($(V1)+(0.68,0)$) {};
\node at ($(V2)+(0.5,0)$) {\dots};

\draw[-, thick] (U0) -- ($(V0.north) + (0, 3.5mm)$);
\draw[-, dashed] (U0) -- ($(V1.north) + (0, 3.5mm)$);
\draw[-, thick] (U0) -- ($(V2.north) + (0, 3.5mm)$);

\node at (0,-0.2) {\footnotesize U(G)};
\node[whitedot, label=above:{\footnotesize$u_0$}] (U0b) at (0.5,-0.2) {};
\node[whitedot, label=above:{\footnotesize$u_1$}] (U1b) at ($(U0b)+(0.68,0)$) {};
\node[whitedot, label=above:{\footnotesize$u_2$}] (U2b) at ($(U1b)+(0.68,0)$) {};
\node[whitedot, label=above:{\footnotesize$u_3$}] (U3b) at ($(U2b)+(0.68,0)$) {};
\node[whitedot, label=above:{\footnotesize$u_4$}] (U4b) at ($(U3b)+(0.68,0)$) {};
\node[whitedot, label=above:{\footnotesize$u_5$}] (U5b) at ($(U4b)+(0.68,0)$) {};

\draw[-, thick] (V0) -- ($(U0b.north) + (0, 3.5mm)$);
\draw[-, thick] (V0) -- ($(U1b.north) + (0, 3.5mm)$);
\draw[-, dashed] (V0) -- ($(U2b.north) + (0, 3.5mm)$);
\draw[-, thick] (V0) -- ($(U3b.north) + (0, 3.5mm)$);
\draw[-, dashed] (V0) -- ($(U4b.north) + (0, 3.5mm)$);
\draw[-, thick] (V0) -- ($(U5b.north) + (0, 3.5mm)$);

\draw[-, dashed] (V1) -- ($(U0b.north) + (0, 3.5mm)$);
\draw[-, thick] (V1) -- ($(U1b.north) + (0, 3.5mm)$);
\draw[-, thick] (V1) -- ($(U3b.north) + (0, 3.5mm)$);

\draw[-, thick] (V2) -- ($(U0b.north) + (0, 3.5mm)$);
\draw[-, thick] (V2) -- ($(U1b.north) + (0, 3.5mm)$);
\draw[-, thick] (V2) -- ($(U4b.north) + (0, 3.5mm)$);
\draw[-, dashed] (V2) -- ($(U5b.north) + (0, 3.5mm)$);

\node[anchor=west] at (4.7,3.7) {$\mathbf{Block_0}$};

\begin{scope}[shift={(4.7,2.8)}]
  \foreach \i/\x in {0/0,1/0.5,2/1.0}{
    \node[draw, minimum width=0.3cm, minimum height=0.3cm, fill=gray] (blk\i) at (\x,0) {};
  }

\node at (6,-0.1) {\footnotesize $N_G(u_1)$};
\node at (7.5,-2.0) {\footnotesize $N_G^2(u_1)$};

  \node[draw, minimum width=0.3cm, minimum height=0.3cm, fill=gray] (blkLast) at (2.0,0) {};

  \node at ($(blk2)!0.5!(blkLast)$) {\footnotesize$\dots$};

\node at ($(blk2)!0.5!(blkLast)+(0,-0.90)$) {\footnotesize$\dots$};

\path[draw, ultra thick]
  (blk0.south)
  .. controls +(0.25,-0.5) and +(-0.25,0.5) .. ++(0,-0.8)
  .. controls +(0.25,-0.5) and +(-0.25,0.5) .. ++(0,-0.8)
  coordinate (T0end);

\node[below=-2pt of T0end, anchor=west] {\tiny T$_0$};

\path[draw, thin]
  (blk1.south)
  .. controls +(0.25,-0.5) and +(-0.25,0.5) .. ++(0,-0.4)
  .. controls +(0.25,-0.5) and +(-0.25,0.5) .. ++(0,-0.4)
  coordinate (T1end);

\node[below=-2pt of T1end, anchor=west] {\tiny T$_1$};

\path[draw, line width=1.0pt]
  (blk2.south)
  .. controls +(0.25,-0.5) and +(-0.25,0.5) .. ++(0,-0.65)
  .. controls +(0.25,-0.5) and +(-0.25,0.5) .. ++(0,-0.65)
  coordinate (T2end);

\node[below=-2pt of T2end, anchor=west] {\tiny T$_2$};

\path[draw,  thin]
  (blkLast.south)
  .. controls +(0.25,-0.5) and +(-0.25,0.5) .. ++(0,-0.5)
  .. controls +(0.25,-0.5) and +(-0.25,0.5) .. ++(0,-0.5)
  coordinate (TNend);

\node[below=-2pt of TNend, anchor=west] {\tiny T$_N$};

\node[draw, rounded corners, fit=(blk0)(blkLast), inner sep=3pt, label=above:{\footnotesize Threads}] (threadbox) {};

\node[fill=gray!50, rounded corners, 
      fit=(blk0)(blkLast), 
      minimum height=0.3cm,  
      below=1.7cm of blk2] (barrier) {\footnotesize Intermediate buffer};;

\node[draw, fill=white, rounded corners,
      fit=(blk0)(blkLast), 
      minimum height=0.3cm, 
      below=0.35cm of blk2] at (barrier) {\footnotesize Synchronization};

\end{scope}

\draw[dashed, thick] 
  ($(blkLast.east)+(0.4,0.6)$) -- 
  ++(0,-4); 

\begin{scope}[shift={(8,0)}] 
\node[anchor=west] at (0,3.7) {$\mathbf{Block_1}$};

\node at (0,2.8) {\footnotesize U(G)};
\node[whitedot, label=above:{\footnotesize$u_1$}] (U1) at (0.5,2.8) {};

\node at (0,1.5) {\footnotesize V(G)};
\node[blackdot, label=above:{\footnotesize$v_0$}] (V1) at (0.5,1.5) {};
\node[blackdot, label=above:{\footnotesize$v_1$}] (V2) at ($(V1)+(0.68,0)$) {};
\node[blackdot, label=above:{\footnotesize$v_2$}] (V3) at ($(V2)+(0.68,0)$) {};
\node[blackdot, label=above:{\footnotesize$v_3$}] (V4) at ($(V3)+(0.68,0)$) {};
\node[blackdot, label=above:{\footnotesize$v_4$}] (V5) at ($(V4)+(0.68,0)$) {};
\node[blackdot, label=above:{\footnotesize$v_5$}] (V6) at ($(V5)+(0.68,0)$) {};
\node[blackdot, label=above:{\footnotesize$v_6$}] (V7) at ($(V6)+(0.68,0)$) {};

\draw[-, thick] (U1) -- ($(V1.north)+(0,3.5mm)$);
\draw[-, dashed] (U1) -- ($(V2.north)+(0,3.5mm)$);
\draw[-, thick] (U1) -- ($(V3.north)+(0,3.5mm)$);
\draw[-, dashed] (U1) -- ($(V4.north)+(0,3.5mm)$);
\draw[-, thick] (U1) -- ($(V5.north)+(0,3.5mm)$);
\draw[-, dashed] (U1) -- ($(V6.north)+(0,3.5mm)$);
\draw[-, thick] (U1) -- ($(V7.north)+(0,3.5mm)$);

\node at (0,-0.2) {\footnotesize U(G)};
\node[whitedot, label=above:{\footnotesize$u_0$}] (U1) at (0.5,-0.2) {};
\node[whitedot, label=above:{\footnotesize$u_1$}] (U2) at ($(U1)+(0.68,0)$) {};
\node[whitedot, label=above:{\footnotesize$u_2$}] (U3) at ($(U2)+(0.68,0)$) {};
\node[whitedot, label=above:{\footnotesize$u_3$}] (U4) at ($(U3)+(0.68,0)$) {};


\draw[-, thick]   (V1) -- ($(U1.north)+(0,3.5mm)$);
\draw[-, dashed]  (V1) -- ($(U2.north)+(0,3.5mm)$);

\draw[-, thick]   (V2) -- ($(U1.north)+(0,3.5mm)$);
\draw[-, thick]   (V2) -- ($(U2.north)+(0,3.5mm)$);
\draw[-, thick]   (V2) -- ($(U3.north)+(0,3.5mm)$);
\draw[-, dashed]  (V2) -- ($(U4.north)+(0,3.5mm)$);

\draw[-, thick]   (V3) -- ($(U3.north)+(0,3.5mm)$);

\draw[-, thick]   (V4) -- ($(U1.north)+(0,3.5mm)$);
\draw[-, thick]   (V4) -- ($(U2.north)+(0,3.5mm)$);

\draw[-, dashed]  (V5) -- ($(U4.north)+(0,3.5mm)$);

\draw[-, thick]   (V6) -- ($(U1.north)+(0,3.5mm)$);
\draw[-, thick]   (V6) -- ($(U2.north)+(0,3.5mm)$);
\draw[-, thick]   (V6) -- ($(U3.north)+(0,3.5mm)$);
\draw[-, thick]   (V6) -- ($(U4.north)+(0,3.5mm)$);

\draw[-, dashed]  (V7) -- ($(U1.north)+(0,3.5mm)$);

\end{scope}

\begin{scope}[shift={(13.4,0)}] 
\node[anchor=west] at (0,3.7) {$\mathbf{Block_2} \ \dots$};

\node at (0,2.8) {\footnotesize U(G)};
\node[whitedot, label=above:{\footnotesize$u_2$}] (U2top) at (0.5,2.8) {};

\node at (3.0,2.5) {\footnotesize $N_G(u_2)$};
\node at (3.5,1.0) {\footnotesize $N_G^2(u_2)$};

\node at (0,1.5) {\footnotesize V(G)};
\node[blackdot, label=above:{\footnotesize$v_3$}] (V1b) at (0.5,1.5) {};
\node[blackdot, label=above:{\footnotesize$v_4$}] (V2b) at ($(V1b)+(0.68,0)$) {};
\node[blackdot, label=above:{\footnotesize$v_5$}] (V3b) at ($(V2b)+(0.68,0)$) {};
\node[blackdot, label=above:{\footnotesize$v_6$}] (V4b) at ($(V3b)+(0.68,0)$) {};
\node[blackdot, label=above:{\footnotesize$v_7$}] (V5b) at ($(V4b)+(0.68,0)$) {};

\draw[-, thick] (U2top) -- ($(V1b.north)+(0,3.5mm)$);
\draw[-, dashed] (U2top) -- ($(V2b.north)+(0,3.5mm)$);
\draw[-, thick] (U2top) -- ($(V3b.north)+(0,3.5mm)$);
\draw[-, dashed] (U2top) -- ($(V4b.north)+(0,3.5mm)$);
\draw[-, thick] (U2top) -- ($(V5b.north)+(0,3.5mm)$);

\node at (0,-0.2) {\footnotesize U(G)};
\node[whitedot, label=above:{\footnotesize$u_0$}] (U0low) at (0.5,-0.2) {};
\node[whitedot, label=above:{\footnotesize$u_1$}] (U1low) at ($(U0low)+(0.68,0)$) {};
\node[whitedot, label=above:{\footnotesize$u_3$}] (U3low) at ($(U1low)+(0.68,0)$) {};
\node[whitedot, label=above:{\footnotesize$u_4$}] (U4low) at ($(U3low)+(0.68,0)$) {};
\node[whitedot, label=above:{\footnotesize$u_5$}] (U5low) at ($(U4low)+(0.68,0)$) {};

\draw[-, thick] (V1b) -- ($(U0low.north)+(0,3.5mm)$);
\draw[-, thick] (V1b) -- ($(U5low.north)+(0,3.5mm)$);

\draw[-, dashed] (V2b) -- ($(U1low.north)+(0,3.5mm)$);
\draw[-, dashed] (V2b) -- ($(U4low.north)+(0,3.5mm)$);
\draw[-, thick] (V2b) -- ($(U5low.north)+(0,3.5mm)$);

\draw[-, thick] (V3b) -- ($(U0low.north)+(0,3.5mm)$);
\draw[-, thick] (V3b) -- ($(U3low.north)+(0,3.5mm)$);
\draw[-, dashed] (V3b) -- ($(U4low.north)+(0,3.5mm)$);
\draw[-, thick] (V3b) -- ($(U5low.north)+(0,3.5mm)$);

\draw[-, thick] (V4b) -- ($(U0low.north)+(0,3.5mm)$);
\draw[-, dashed] (V4b) -- ($(U4low.north)+(0,3.5mm)$);
\draw[-, dashed] (V4b) -- ($(U5low.north)+(0,3.5mm)$);

\draw[-, thick] (V5b) -- ($(U1low.north)+(0,3.5mm)$);
\draw[-, thick] (V5b) -- ($(U3low.north)+(0,3.5mm)$);
\draw[-, thick] (V5b) -- ($(U4low.north)+(0,3.5mm)$);

\end{scope}
\end{tikzpicture}
\caption{Example of workload variance across GPU blocks and threads. }
\label{fig:workload}
\end{figure*}

\section{Optimization in GPU: Load balancing}
\label{sec:load}
 In this section, we present our design to address the problem of load imbalance with static scheduling and intra-block load imbalance.

 In the earlier version of our GPU algorithm, we adopted a static scheduling strategy in which each GPU block was permanently assigned to process a fixed vertex from the designated partition. Although simple to implement, this approach leads to severe load imbalance because the degrees of vertices in real-world bipartite graphs are highly skewed. Vertices with small degrees finish quickly, causing their blocks to remain idle while other blocks continue processing vertices with substantially larger neighborhoods. As a result, a large portion of the GPU's streaming multiprocessors (SMs) stay underutilized, limiting overall throughput.

 To address this, we adopt a dynamic scheduling model based on persistent thread blocks. Instead of launching one block for every vertex, we allocate only a small set of long-lived blocks, typically proportional to the number of streaming multiprocessors, and these blocks remain active throughout the computation. Each block repeatedly fetches a new vertex to process from a global atomic task counter. As soon as the block completes its current vertex, it immediately retrieves another vertex without waiting for the other blocks. This persistent execution model allows the GPU to redistribute work automatically, ensuring high hardware occupancy even when the input graph contains vertices whose degrees differ by several orders of magnitude. The dynamic algorithm does not store its wedge counters or two-hop bitmaps in shared memory because shared memory is too small to accommodate the full two-hop destination space of a vertex, which often spans tens to hundreds of thousands of vertices in  real bipartite graphs. While our static tiled version uses shared memory effectively by processing it in small tiles, dynamic scheduling repeatedly assigns new vertices to blocks, making per-task initialization of large shared-memory tiles costly and inefficient. To maintain high SM occupancy and avoid repeated shared-memory resets, the dynamic kernel allocates its per-block buffers in global memory, while shared memory is used only for lightweight block-level reduction. This enables performance over diverse degree distributions and boosts runtime stability under dynamic load balancing.

Balanced butterfly counting in large-scale signed bipartite networks raises significant computational challenges due to the highly uneven degree distributions typical of real-world networks. In Algorithm~\ref{algo:gpu-bfc}, each GPU block is responsible for processing a particular vertex, with the adjacency lists varying widely in size. Within each block, threads deal with different divisions of the adjacency list assigned to the vertex. For instance, as demonstrated in Fig.~\ref{fig:workload}, blocks 0, 1, and 2 have varying workloads, mostly influenced by the degrees of their allocated vertices. However, because the sizes of the adjacency lists can vary considerably, the computational time per block becomes uneven. Therefore, the computational workload across threads is highly uneven, resulting in significant intra-block load imbalance. To fill this, we employ a hierarchical workload distribution strategy that dynamically selects the most appropriate GPU kernel for each vertex \cite{xia2024gpu}. The workload of a vertex is determined by two key metrics: its degree and its two-hop neighborhood size, both of which directly influence the computational complexity of butterfly counting.

\begin{algorithm}[t!]
\small
\SetKwInOut{Input}{Input}
\SetKwInOut{Output}{Output}
\Input{$G(U,V,E,\sigma)$: A signed bipartite graph.}
\Output{$B(G)$: Total number of $\balancebutterfly$.}
\BlankLine

\BlankLine
\tcp*[f]{Host-side preprocessing:}
Select the smaller partition $S=\min(U,V)$\;
\ForEach{$u \in S$}{
    $\text{fanout}(u) \gets \sum_{v \in N(u)} \deg(v) + \deg(u)$
}
Sort vertices in $S$ in decreasing order of $\text{fanout}(u)$\;
\tcp*[f]{GPU kernel execution:}
\BlankLine
$\balancebutterfly \gets 0$ \\
$\texttt{nextTask} \gets 0$ \tcp*[f]{Dynamic scheduling}

\While{true}{ 
    \If{\textnormal{threadIdx.x} = 0}{
        $taskIdx \gets atomicAdd(\texttt{nextTask}, 1)$
    }
    $\texttt{syncthreads}()$

    \If{$taskIdx \ge |S|$}{ \textbf{break} }

    $u \gets S[taskIdx]$ \\


    \If{$\big|N_G(u)\cap N_G^{2}(u)\big| < threshold$}{
    use warp-level cooperation\;
}
\Else{
    use block-level cooperation\;
}
    Initialize $B_1$ and $B_2$ to $0$ \\   
    \ForEach{$v \in \Gamma(u)$}{
        \ForEach{$w \in \Gamma(v)$}{
            \If{$w.\text{id} > u.\text{id}$}{
                \If{$\sigma(u,v) = \sigma(v,w)$}{
                    $B_1[w] \gets B_1[w] + 1$
                }
                \Else{
                    $B_2[w] \gets B_2[w] + 1$
                }
            }
        }
    }

    \BlankLine
    $\balancebutterfly \gets$ $reduction$($ \binom{B_1[w]}{2} + \binom{B_2[w]}{2}$) \\
    
    $\texttt{syncthreads}()$
}

\Return $\balancebutterfly$

\caption{G-BBC++: GPU Dynamic scheduling Balanced Butterfly Counting}
\label{algo:gpu-bfc1}
\end{algorithm}

\subsection{G-BBC++ Algorithm} 
 We present the G-BBC++ pseudocode for dynamic load balance in Algorithm \ref{algo:gpu-bfc1}. We utilize dynamic scheduling, adaptive intra-block parallelism, and block-local workspaces because they eliminate load imbalance, reduce memory overhead from millions of entries to only those corresponding to active elements, maximize thread utilization, and enable the GPU to process each vertex independently and effectively.
 
 First, we choose the smaller partition $S$ because butterfly counting expands from a vertex $u$ through its neighbors and their neighbors, so starting from the smaller side minimizes the number of two-hop traversals. Next, we compute a fanout score for each $u \in S$ and sort the vertices by decreasing the fanout value so that heavy vertices are processed early by fully active thread blocks. This prevents situations where some thread blocks finish quickly on low-degree vertices while others remain stuck on very high-degree vertices, thus improving overall GPU utilization. 
 
 Dynamic scheduling using a global counter is essential because vertices in real bipartite graphs have highly skewed degrees. Static one-thread-block-per-vertex assignment would cause severe load imbalance. Instead, each block repeatedly fetches a new vertex to process after completing its current one, ensuring that all streaming multiprocessors (SMs) are continuously busy regardless of degree skew.
 
 Each block allocates its own private counters $B_1$ and $B_2$ (for storing symmetric and asymmetric wedges). We use intra-block parallelism because $|N_G(u)|$ and $|N_G^{2}(u)|$ vary widely across vertices. If all threads were always used, small-degree vertices would spend most of their time idle, resulting in poor performance. Conversely, if only one warp were used for all vertices, large-degree vertices would become bottlenecks. Therefore, we employ three execution regimes: warp-only, partial-block, and full-block cooperation, based on the degree and 2-hop degree of the processed vertex $u$. Specifically, vertices whose degree is < $32$ are handled by a single warp, vertices whose degree lies between $32 \le 512$ use partial-block cooperation, and vertices with degree $>512$ are processed using the full block. This technique offers high efficiency across highly biased degree distributions. During the 2-hop exploration, each thread performs atomic updates to $B_1[w]$ and $B_2[w]$, since multiple threads may encounter the same $w$. The use of atomic operations offers correctness under high parallelism. After the enumeration, the algorithm aggregates and performs $\binom{B_1[w]}{2}$ + $\binom{B_2[w]}{2}$ for each $w$. A warp-shuffle-based reduction is employed as it avoids global or shared-memory barriers and offers highly efficient intra-block summation.

\subsection{Complexity Analysis} 
The algorithm first selects the smaller side of the two partitions, $S = \min(U, V)$, in constant time. For each vertex $u \in S$, it computes a fanout score $\text{fanout}(u)=\sum_{v \in \Gamma(u)} \deg(v) + \deg(u)$, which requires $O(\sum_{u \in S} \deg(u))$ time in total. The vertices in $S$ are then sorted in decreasing order of fanout, incurring a cost of $O(|S| \log |S|)$. This preprocessing is performed once on the host and is asymptotically dominated by the subsequent GPU computation for large graphs.

The dominant cost arises from the GPU kernel, where vertices $u \in S$ are processed using dynamic scheduling with persistent thread blocks. For a fixed vertex $u$, the kernel enumerates all two-hop paths $(u,v,w)$ by iterating over all neighbors $v \in \Gamma(u)$ and, for each such neighbor, all neighbors $w \in \Gamma(v)$. Consequently, the work required to process a single vertex $u$ is proportional to $\sum_{v \in \Gamma(u)} \deg(v)$. The use of adaptive intra-block parallelism, such as warp- or block-level cooperation, affects only constant factors and does not change the asymptotic work.
By summing over all vertices in the selected partition $S$, the total work complexity is $O(\sum_{u \in S} \sum_{v \in \Gamma(u)} \deg(v))$. Dynamic scheduling ensures that this work is evenly distributed across GPU blocks, mitigating load imbalance caused by highly skewed degree distributions. As a result, the parallel running time is determined by this total work divided among the available GPU blocks, while the asymptotic complexity remains unchanged.


\section{Experimental Evaluation}
\label{Sec:EE}
In this section, we assess the performance of the proposed algorithms, M-BBC and G-BBC, and compare them with the CPU-based BB2K and SBCList++ methods.

\textbf{Computing Resources}. We implemented all the algorithms in C++. Experiments were conducted on a workstation with an Intel Xeon E5-2697 v3 dual-socket CPU (28 physical cores, 56 threads, 2.60~GHz), 256~GB RAM, and an NVIDIA Quadro RTX~6000 GPU (24~GB, CUDA~12.9, Driver~575.51.03), running Ubuntu 64-bit. We use the \texttt{Intel TBB} library to enable parallelism in our implementation. Specifically, we employ \texttt{tbb::parallel\_for} to parallelize loops, \texttt{tbb::concurrent\_hash\_map} for thread-safe bucket operations, and \texttt{tbb::atomic} for synchronized updates to the balanced butterfly count.

\textbf{Algorithms:} We evaluate the performance of the following algorithms.

\begin{itemize}
    \item \textbf{SBCList++}:
    An adaptation of BCList++~\cite{yang2023p}, originally designed for enumerating $(p,q)$-bicliques in unsigned bipartite graphs. We extend and enhance it to support counting balanced $(2,2)$-bicliques in \emph{signed} bipartite graphs.

    \item \textbf{BB2K}:
    A serial baseline algorithm for counting balanced $(2,2)$-bicliques, proposed in our earlier work~\cite{kiran2024efficient}.

    \item \textbf{M-BBC}:
    The proposed multi-core parallel algorithm for counting balanced butterflies, as described in Section~\ref{Sec:alg}.

    \item \textbf{G-BBC}:
    The proposed GPU-based parallel algorithm for balanced butterfly (i.e., $(2,2)$-biclique) counting, also described in Section~\ref{Sec:alg}.

    \item \textbf{G-BBC++}:  
An enhanced version of G-BBC that incorporates load balancing and dynamic scheduling to improve workload distribution across GPU threads and reduce imbalance, as described in Section~\ref{sec:load}.
    
\end{itemize}


\subsection{Datasets Description}

\begin{table}[!ht]
\small
\renewcommand{\arraystretch}{1.2}
\caption{Characteristics of the Datasets.}
\label{table:datasets}

\resizebox{\columnwidth}{!}{
\begin{tabular}{|c|c|c|c|}
\hline
\textbf{Dataset} & $|U|$ & $|V|$  & $|E|$ \\
\hline
Senate (\textbf{SE})        & $145$      & $1,201$      & $27,083$      \\
\hline
DBLP (\textbf{DBLP})        & $6,001$    & $1,308$      & $29,256$      \\
\hline
House (\textbf{HO})         & $515$      & $1,281$      & $114,378$     \\ 
\hline
Wiki-Nap (\textbf{NAP})     & $1,753$    & $25,881$     & $265,546$     \\ \hline
BookCrossing (\textbf{BC})  & $77,802$   & $185,955$    & $433,652$     \\ \hline
NIPS Papers (\textbf{NIPS}) & $1,500$    & $12,375$     & $746,315$     \\ \hline
Last.fm (\textbf{LF})       & $992$      & $174,077$    & $898,062$     \\ \hline
Movielens (\textbf{MV})     & $6,040$    & $3,706$      & $1,000,208$   \\ \hline
Jetser 150 (\textbf{JE})    & $50,692$   & $140$        & $1,728,847$   \\ \hline
KDD Cup (\textbf{KDD})      & $255,170$  & $1,848,114$  & $2,766,393$   \\ \hline
Digg Votes (\textbf{DG})    & $139,409$  & $3,553$      & $3,010,197$   \\ \hline
AOL (\textbf{AOL})          & $4,811,647$ & $1,632,788$ & $10,741,953$  \\ \hline
Epinions (\textbf{EP})      & $120,492$  & $755,760$    & $13,668,320$  \\ \hline
Netflix (\textbf{Netflix})  & $480,189$  & $17,770$     & $100,480,507$ \\ \hline
Yahoo (\textbf{Yahoo})      & $1,000,990$ & $624,961$   & $256,804,235$ \\
\hline
\end{tabular}
}
\end{table}

We evaluate our proposed algorithms using a diverse collection of $15$ real-world bipartite datasets, encompassing both sparse and dense graphs. The \textbf{SE} and \textbf{HO} datasets are obtained from the signed bipartite repository\footnote{\url{https://github.com/tylersnetwork/signed_bipartite_networks}}. The \textbf{DBLP}, \textbf{MV}, \textbf{KDD}, \textbf{AOL}, and \textbf{Netflix} datasets are sourced from the bipartite network repository\footnote{\url{https://renchi.ac.cn/datasets/}}. The remaining datasets, \textbf{NAP}, \textbf{BC}, \textbf{NIPS}, \textbf{LF}, \textbf{JE}, \textbf{DG}, \textbf{EP}, and \textbf{Yahoo}, are obtained from the KONECT collection\footnote{\url{http://konect.cc/networks/}}. 

Among these datasets, SE and HO are originally provided as \emph{signed bipartite graphs}. The others are supplied as unsigned networks, some of which include rating information that implicitly encodes positive or negative interactions. To ensure a consistent, comparable notion of polarity across all datasets, we transform the rating-based graphs into signed bipartite graphs by assigning binary sign labels to their rating values. For instance, in the Jester (JE) dataset, which uses a $10$-star rating system, ratings strictly greater than $6$ are assigned a value of $1$, while all other ratings are mapped to $0$. Similarly, in the MV, EP, and Netflix datasets, which use a $5$-star rating system, ratings of four and above are categorized as $1$, and lower ratings are categorized as $0$. This normalization ensures consistent polarity interpretation across datasets with heterogeneous rating schemes. The unsigned datasets that do not contain rating information, specifically DBLP, BC, NAP, NIPS, LF, KDD, DG, AOL, and Yahoo, are converted into synthetic signed bipartite graphs following the methodology in~\cite{b2}. Each edge in these networks is assigned a positive label $(1)$ with probability $0.7$ and a negative label $(0)$ with probability $0.3$, producing realistic polarity distributions for evaluating balanced butterfly detection. 

Several datasets, such as DG, NAP, and LF, contain duplicate edges between the same pair of nodes, often with conflicting signs (an edge might appear multiple times with different signs). To fix these inconsistencies, we perform a preprocessing step in which, for each duplicated edge, we keep the most recent interaction based on its timestamp.
A complete summary of datasets is presented in Table~\ref{table:datasets}.

\subsection{Performance Assessment}
In this part, we evaluate the performance of the proposed algorithms on the earlier-mentioned datasets, comparing them with the baseline algorithms in terms of processing speed.
















\subsubsection{Performance of BB2K vs SBCList++}

We initially compare our \bbtwok, serial implementation over the adapted baseline SBCList++. As presented in Table~\ref{Tab:BB2k-SBCList++}, the \bbtwok algorithm exhibits notable performance improvements over a wide range of datasets. For instance, on the NAP, LF, and JE datasets, it is often superior to SBCList++ by considerable margins. In particular, on the NAP dataset, SBCList++ takes $356$ seconds to count balanced butterflies, whereas \bbtwok completes in $2.7$ seconds, achieving a speedup of $119.16\times$. Similarly, on the JE dataset, SBCList++ requires $6547$ seconds while \bbtwok in $55.95$ seconds. The primary reason for this improvement lies in the design of the two algorithms. Although SBCList++ correctly counts balanced butterflies, it explicitly enumerates all candidate $(2,2)$-bicliques and verifies their balance for each. This results in a large number of redundant computations, especially in graphs with dense neighborhoods or high-degree vertices. In contrast, \bbtwok separates symmetric and asymmetric wedges into distinct buckets, ensuring that any biclique within a bucket is inherently balanced. This eliminates additional balance checks and reduces the runtime.

However, \bbtwok performs slower on very sparse graphs because its bucket-based wedge-grouping technique becomes less effective when vertex degrees are extremely low. Sparse graphs contain very few wedges and almost no shared neighborhoods, resulting in very small bucket sizes. Consequently, the overhead of building and maintaining the bucket structures dominates the runtime. In contrast, SBCList++ benefits from sparsity, as the number of candidate $(2,2)$-bicliques is very small, making its exhaustive enumeration relatively inexpensive. Thus, although SBCList++ never surpasses \bbtwok, it performs competitively on sparse datasets such as BC, KDD, and AOL due to the reduced enumeration cost and lower constant-time overhead. These findings demonstrate that \bbtwok delivers substantial performance gains over SBCList++, achieving faster execution times across diverse graphs and maintaining competitive performance even on highly sparse datasets.

\begin{table}[!ht]
\small
\renewcommand{\arraystretch}{1.2} 
\caption{Execution times (in seconds) of BB2K and SBCList++.}
\begin{center}
\begin{tabular}{|c|c|c|c|}
\hline
\textbf{Dataset} & \textbf{ $|\balancebutterfly|$} & \textbf{SBCList++} & \textbf{BB2K (speedup)} \\
\hline
\textbf{SE} & $15.32$M & $4.92$ & $0.522$ (\textbf{9.43$\times$}) \\
\hline
\textbf{DBLP} & $0.85$M & $1.61$ & $0.515$ (\textbf{3.13$\times$}) \\
\hline
\textbf{HO} & $280.79$M & $90.83$ & $6.01$ (\textbf{14.96$\times$}) \\
\hline
\textbf{NAP} & $2.24$B & $356.3$ & $2.99$ (\textbf{119.16$\times$}) \\
\hline
\textbf{BC} & $1.11$M & $5.454$ & $3.537$ (\textbf{1.54$\times$}) \\
\hline
\textbf{NIPS} & $3.75$B & $1845.35$ & $172$ (\textbf{10.72$\times$}) \\
\hline
\textbf{LF} & $2.35$B & $2363.125$ & $58.19$ (\textbf{40.61$\times$})\\
\hline
\textbf{MV} & $8.90$B & $3454.13$ & $240.65$ (\textbf{14.35$\times$}) \\
\hline
\textbf{JE} & $136.88$B & $6547.46$ & $56.12$ (\textbf{116.66$\times$})\\
\hline
\textbf{KDD} & $9.20$M & $16$ & $12.89$ (\textbf{1.24$\times$})\\
\hline
\textbf{DG} & $15.06$B & $5320$ & $693.89$ (\textbf{7.66$\times$}) \\
\hline
\textbf{AOL} & $104.47$M & $575$ & $109.89$ (\textbf{5.27$\times$}) \\
\hline
\textbf{EP} & $158.33$B & $11800$ & $600.89$  (\textbf{19.6$\times$})\\
\hline
\textbf{Netflix} & $8.39$T & $>10$ hrs & $>10$ hrs \\
\hline
\textbf{Yahoo} & $5.20$T & $>10$ hrs & $>10$ hrs \\
\hline
\end{tabular}
\end{center}
\label{Tab:BB2k-SBCList++}
\end{table}

  \begin{figure}[!htbp]
    \centering
    \includegraphics[width=8.6cm]{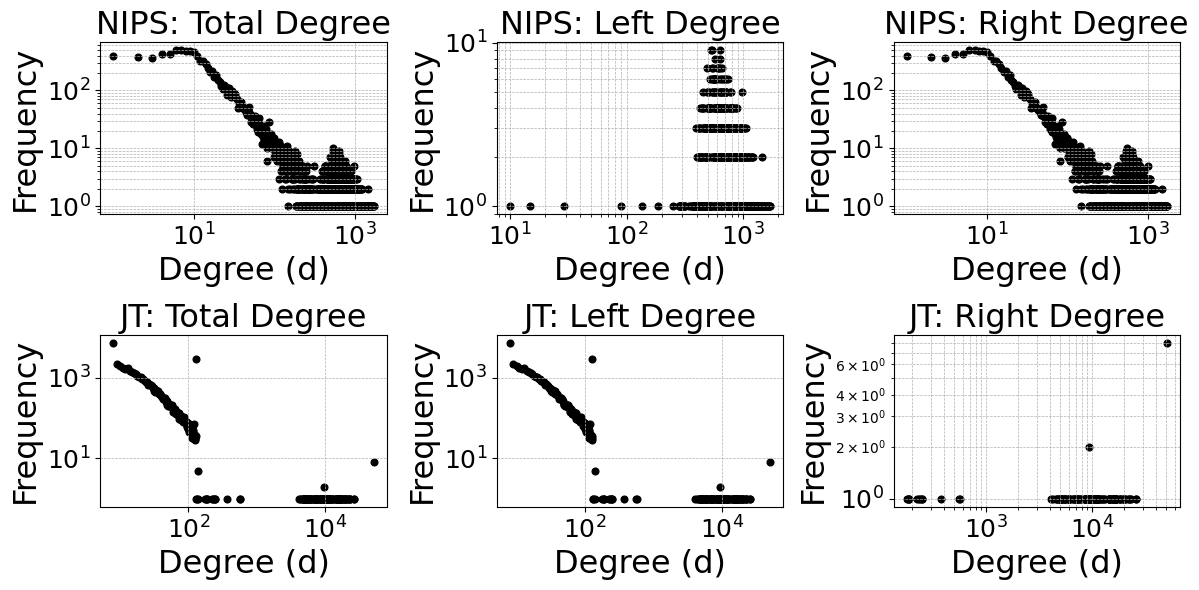}
    \caption{Degree distribution of the NIPS and JE datasets.}
    \label{fig:enter-label1}
\end{figure}

We observe that although the NIPS dataset contains fewer edges than the JE dataset, it incurs a higher runtime in both SBCList++ and \bbtwok. This behavior can be explained by examining the structural properties of the two graphs, particularly the size of their minimum-side partition and the distribution of vertex degree. In the NIPS dataset, the minimum-side partition consists of $1,500$ vertices, whereas in the JE dataset, this partition contains only $150$ vertices. Since both SBCList++ and \bbtwok process vertices on the minimum-side partition as the outer loop of their computation, the cost of the algorithm is proportional to the number of vertices on this side and to the amount of wedge expansion that each vertex induces. Consequently, the tenfold increase in the number of minimum-side vertices in NIPS directly leads to substantially more wedge-generating operations than in JE. Furthermore, the NIPS dataset exhibits a highly skewed and heavy-tailed degree distribution. Several vertices on the minimum side have very high degrees, which significantly increases the number of two-hop paths that must be examined during wedge formation. These high-degree vertices dominate the computational cost due to the quadratic growth in neighborhood intersections. In contrast, the JE dataset has a more uniform degree distribution on its minimum side, resulting in fewer and more evenly distributed wedges. As a result, although JE contains more edges overall, its more favorable structural characteristics lead to faster execution. Figure~\ref{fig:enter-label1} shows the detailed degree distributions of the NIPS and JE datasets.

\subsubsection{Performance of BB2K vs M-BBC} 
 Although $\bbtwok$ is an efficient CPU-based algorithm for counting balanced motifs such as $(2,2)$-bicliques, its computational cost becomes a bottleneck for large-scale datasets. For example, the Netflix and Yahoo datasets contain approximately 100 million and 250 million edges, respectively, and \bbtwok exceed our maximum runtime threshold of ten hours on these inputs. In practice, even well-optimized serial algorithms struggle to meet the latency requirements of real-time or high-throughput applications. By considering that, we go beyond $\bbtwok$ to a multicore design, called M-BBC. M-BBC parallelizes the wedge enumeration for each vertex across multiple CPU cores, considerably accelerating the computation while maintaining the structural correctness of the original algorithm.

\begin{table}[!htbp]
\small
\renewcommand{\arraystretch}{1.2} 
\caption{Execution times (in seconds) of BB2K and M-BBC.}
\begin{center}
\begin{tabular}{|c|c|c|}
\hline
\textbf{Dataset} & \textbf{BB2K} & \textbf{M-BBC (Speedup)} \\
\hline

\textbf{SE}   & 0.52   & 0.01 (\textbf{52$\times$}) \\ \hline
\textbf{DBLP} & 0.51   & 0.02 (\textbf{25.5$\times$}) \\ \hline
\textbf{HO}   & 6.01   & 0.11 (\textbf{54.63$\times$}) \\ \hline
\textbf{NAP}  & 2.99   & 0.14 (\textbf{21.35$\times$}) \\ \hline
\textbf{BC}   & 3.53   & 0.35 (\textbf{9.22$\times$}) \\ \hline

\textbf{NIPS} & 172.00 & 2.42 (\textbf{71.07$\times$}) \\ \hline
\textbf{LF}   & 58.13  & 1.15 (\textbf{50.54$\times$}) \\ \hline
\textbf{MV}   & 240.65 & 3.52 (\textbf{68.36$\times$}) \\ \hline
\textbf{JE}   & 56.12  & 1.35 (\textbf{41.57$\times$}) \\ \hline

\textbf{KDD}  & 12.89  & 2.45 (\textbf{5.26$\times$})  \\ \hline
\textbf{DG}   & 693.89 & 11.05 (\textbf{62.79$\times$}) \\ \hline
\textbf{AOL}  & 109.89 & 7.55 (\textbf{14.55$\times$}) \\ \hline
\textbf{EP}   & 600.00 & 23.71 (\textbf{25.31$\times$}) \\ \hline

\textbf{Netflix} & $>10$ hrs & 2706.00 \\ \hline
\textbf{Yahoo}   & $>10$ hrs & 9303.53 \\ \hline

\end{tabular}
\end{center}
\label{Tab:BB2k-MBBC}
\end{table}


As shown in Table~\ref{Tab:BB2k-MBBC}, M-BBC constantly improves over the serial \bbtwok algorithm throughout all datasets by efficiently leveraging multi-core parallelism. Specifically, on the NIPS dataset, M-BBC achieves a speedup of $71.07\times$ over \bbtwok. In the same way, on the MV dataset, the running time drops from $240.65$ seconds to just $3.52$, obtaining $68.36\times$ speedup. Even on relatively smaller datasets such as SE and DBLP, in which the parallel work is naturally limited due to small vertex sets and fewer wedges, M-BBC gains substantial speedups of $52\times$ and $25.37\times$, respectively. For moderate datasets such as HO, LF, JE, and DG, the multi-core algorithm exhibits strong scalability, achieving speedups from $41\times$ to $62\times$. The running-time gap becomes more noticeable on very large datasets containing tens of millions of edges. For instance, the \bbtwok algorithm exceeds its 10-hour time limit on the Netflix and Yahoo datasets, indicating that the serial approach is not feasible for large-scale graphs. In contrast, M-BBC completes the computation in approximately $45$ minutes on the Netflix dataset and in 2 hours on the Yahoo dataset, exhibiting its robustness and scalability under high-degree heterogeneity and large wedge spaces. One of the key strengths of M-BBC lies in its inherent vertex-level parallelism, as each vertex $u \in (U \cup V)$ in the graph $G=(U, V, E)$ can be processed independently, allowing massive parallelism. This advantage not only simplifies the parallelization but also ensures that our algorithm scales efficiently with the number of processing cores. 
Next, Fig.~\ref{fig:mbbc-multicore} depicts the execution times of M-BBC on different numbers of cores in seconds. We observe that as the number of cores increases, M-BBC consistently achieves lower runtime across all datasets, indicating strong scalability and additionally highlighting the benefit of vertex-level parallelism in effectively utilizing multi-core processors.

\begin{table*}[!ht]
\small
\renewcommand{\arraystretch}{1.2} 
\caption{Execution times (in seconds) of BB2K, M-BBC, and G-BBC++.}
\begin{center}
\begin{tabular}{|c|c|c|c|c|c|}
\hline
\textbf{Dataset} &
\textbf{BB2K} &
\textbf{M-BBC} &
\textbf{G-BBC++} &
\textbf{Speedup (G-BBC++ / BB2K))} &
\textbf{Speedup (G-BBC++ / M-BBC)} \\
\hline

\textbf{SE}      & $0.52$ & $0.01$ & $0.001$ & \textbf{520$\times$} & \textbf{10$\times$} \\ \hline

\textbf{DBLP}    & $0.51 $& $0.02$ & $0.001$ & \textbf{510$\times$} & \textbf{20$\times$} \\ \hline

\textbf{HO }     & $6.01$ & $0.11 $& $0.003$ & \textbf{2003$\times$} & \textbf{36$\times$} \\ \hline

\textbf{NAP }    & $2.99$ & $0.14$ & $0.007$ & \textbf{427$\times$} & \textbf{20$\times$} \\ \hline

\textbf{BC}      & $3.53$ & $0.35$ & $0.029$ & \textbf{121$\times$}  & \textbf{12$\times$} \\ \hline

\textbf{NIPS}    & $172$ & $2.42$ & $0.013$ & \textbf{13230.81$\times$} & \textbf{186$\times$} \\ \hline

\textbf{LF}      & $58.13$  & $1.15$ & $0.014$ & \textbf{4152.70$\times$} & \textbf{82$\times$} \\ \hline

\textbf{MV}      & $240.65$ & $3.52$ & $0.052$ & \textbf{4627$\times$} & \textbf{67$\times$} \\ \hline

\textbf{JE}      & $56.12$  & $1.35$ & $0.035$ & \textbf{1603$\times$} & \textbf{38$\times$} \\ \hline

\textbf{KDD}  &$ 12.89 $ & $2.45$ & $0.049$ & \textbf{263.80$\times$} & \textbf{50$\times$} \\ \hline

\textbf{DG}      & $693.89$ &$ 11.05$ & $0.133$ & \textbf{5217$\times$} & \textbf{83$\times$} \\ \hline

\textbf{AOL}     & $109.89$ & $7.55$ & $1.001$ & \textbf{109$\times$} & \textbf{7$\times$} \\ \hline

\textbf{EP}      & $600$ & $23.71$ & $0.587$ & \textbf{1022.11$\times$} & \textbf{40$\times$} \\ \hline

\textbf{Netflix} & $>10$ hrs & $2706$ & $40$ & -- & \textbf{67$\times$} \\ \hline

\textbf{Yahoo}   & $>10$ hrs & $9303.53$ & $250$ & -- & \textbf{37$\times$} \\ \hline

\end{tabular}
\end{center}
\label{Tab:Serial-Multicore-GPU}
\end{table*}

To enable a fair comparison across datasets with very different absolute runtimes, we present normalized execution times. For each dataset, runtimes are normalized to the lowest-core runtime, designated as the baseline value of $1.0$, with lower normalized values directly reflecting the performance gains from parallelism. Normalized values below $1.0$ indicate performance gains through parallelism. Normalization of that setup is not feasible for the datasets that do not complete on the baseline core count, such as Yahoo. In these instances, we standardize runtimes to the smallest core count at which execution completes and exclude bars for unfeasible setups.



\begin{figure*}[!ht]
    \centering
    \includegraphics[width=\textwidth]{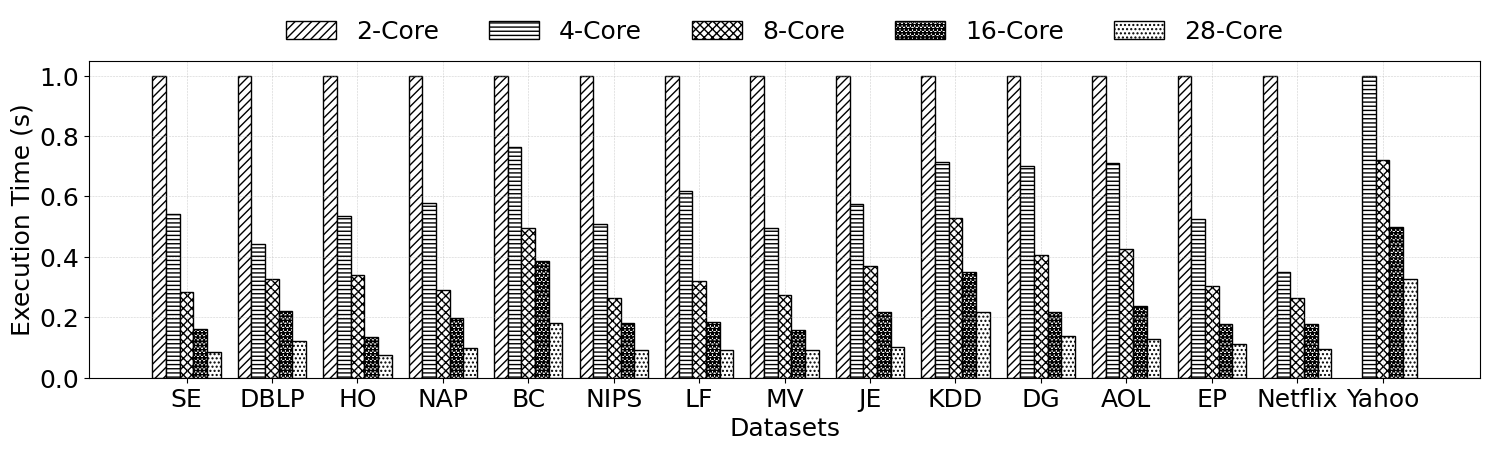}
    \caption{Parallel Speed-ups of M-BBC on Different Cores with Various Datasets.}
    \label{fig:mbbc-multicore}
\end{figure*}

\begin{figure*}[!ht]
    \centering
    \includegraphics[width=\textwidth]{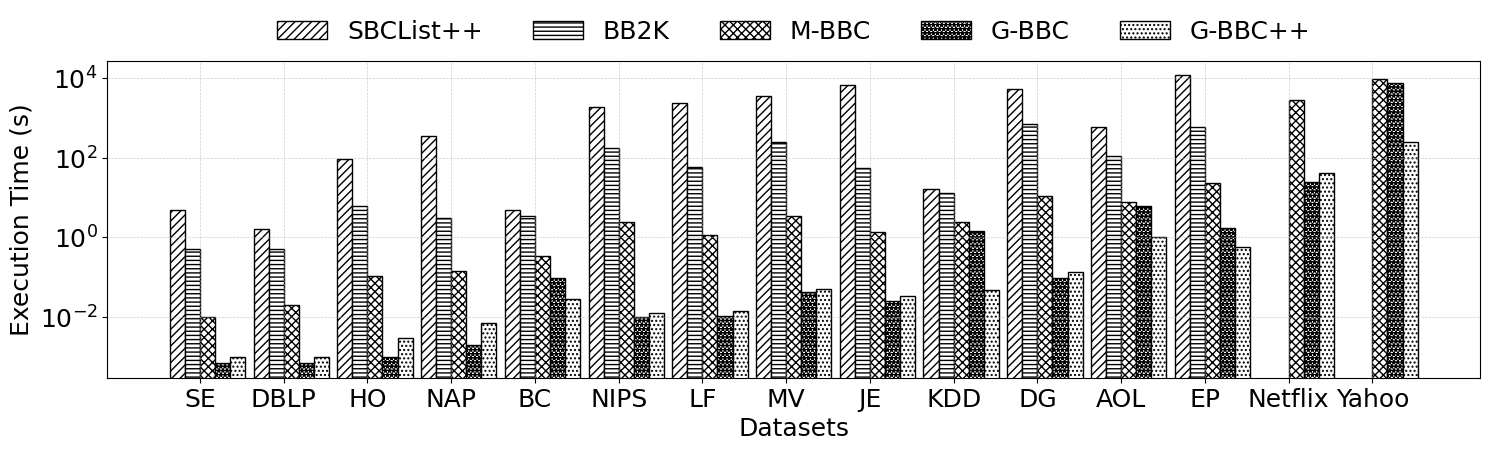}
    \caption{Performance of the proposed algorithm on various datasets.}
    \label{fig:4_alg_comparision bar}
\end{figure*}

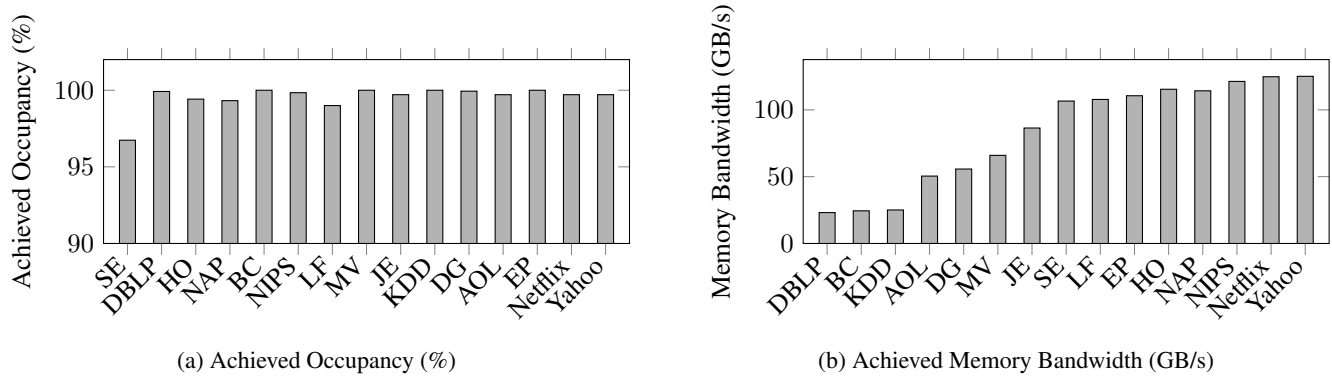
\begin{figure*}[!ht]
\centering

\begin{subfigure}{0.48\linewidth}
\centering
\begin{tikzpicture}
\begin{axis}[
    ybar,
    width=\linewidth,
    height=4.0cm,
    bar width=6pt,
    enlarge x limits=0.05,
    ylabel={Achieved Occupancy (\%)},
    symbolic x coords={ SE, DBLP, HO, NAP,  BC, NIPS, LF, MV, JE,  KDD, DG, AOL, EP, Netflix, Yahoo},
    xtick=data,
    xticklabel style={rotate=45,anchor=east},
    ymin=90, ymax=102
]
\addplot[fill=gray!60] coordinates {
    (SE,96.74)
    (DBLP,99.92)
    (HO,99.42)
    (NAP,99.32)
    (BC,100)
    (NIPS,99.84)
    (LF,99)
    (MV,100)
    (JE,99.71)
    (KDD,100)
    (DG,99.94)
    (AOL,99.71)
    (EP,100)
    (Netflix,99.71)
    (Yahoo,99.71)
};
\end{axis}
\end{tikzpicture}
\caption{Achieved Occupancy (\%)}
\end{subfigure}
\hfill
\begin{subfigure}{0.48\linewidth}
\centering
\begin{tikzpicture}
\begin{axis}[
    ybar,
    width=\linewidth,
    height=4.0cm,
    bar width=6pt,
    enlarge x limits=0.05,
    ylabel={Memory Bandwidth (GB/s)},
    symbolic x coords={ DBLP, BC, KDD, AOL, DG,  MV, JE, SE, LF,   EP,  HO, NAP, NIPS, Netflix, Yahoo},
    xtick=data,
    xticklabel style={rotate=45,anchor=east},
    ymin=0
]
\addplot[fill=gray!60] coordinates {
    (SE,106.60)
    (DBLP,23.12)
    (HO,115.44)
    (NAP,114.27)
    (BC,24.44)
    (NIPS,121.32)
    (LF,107.81)
    (MV,65.95)
    (JE,86.42)
    (KDD,25.08)
    (DG,55.72)
    (AOL,50.44)
    (EP,110.58)
    (Netflix,124.80)
    (Yahoo,125.12)
};
\end{axis}
\end{tikzpicture}
\caption{Achieved Memory Bandwidth (GB/s)}
\end{subfigure}

\caption{GPU Hardware Utilization Across Datasets}
\label{fig:GPU_H/w}
\end{figure*}




\subsubsection{Performance of G-BBC++ vs BB2K and M-BBC }
In the next phase of evaluation, we compare the performance of our GPU-based method, G-BBC++, against both the serial implementation (\bbtwok) and its multicore variant (M-BBC), demonstrating substantial additional acceleration in balanced butterfly counting.

As shown in Table~\ref{Tab:Serial-Multicore-GPU}, in terms of peak performance, G-BBC++ achieves speedups of $13230\times$ over \bbtwok and $186\times$ over M-BBC. On average, G-BBC provides a $2599\times$ improvement relative to \bbtwok and a $50\times$ improvement over M-BBC. The remarkable gains stem from the architectural advantages of GPU computing, dynamic scheduling, and load balancing. Whereas \bbtwok is inherently sequential, and M-BBC is limited to $56$ CPU cores. Fig.~\ref{fig:4_alg_comparision bar} illustrates that our experimental evaluation highlights the substantial performance advantages of our GPU-based algorithm, G-BBC++, over the sequential baseline \(\mathrm{BB2K}\), the multicore implementation \(\mathrm{M\text{-}BBC}\), and the adapted baseline \(\mathrm{SBCList++}\).


\begin{table}[t]
\caption{Structural characteristics of the datasets.}
\label{tab:graph_density_decimal}
\centering
\small
\renewcommand{\arraystretch}{1.2} 
\begin{tabular}{|l|c|c|c|}
\hline
\textbf{Dataset} & $\boldsymbol{n_{\min}}$ & $\boldsymbol{\bar d_{\min}}$ & \textbf{Density} $\boldsymbol{\rho}$ \\
\hline
\textbf{SE }      & 145        & 186.8    & 0.1556 \\ \hline
\textbf{DBLP }    & 1,308      & 22.4     & 0.0037 \\ \hline
\textbf{HO}       & 515        & 222.1    & 0.1732 \\ \hline
\textbf{NAP}      & 1,753      & 151.5    & 0.0059 \\ \hline
\textbf{BC}       & 77,802     & 5.6      &\textbf{ 0.00003} \\ \hline
\textbf{NIPS}     & 1,500      & 497.5    & 0.0402 \\ \hline
\textbf{LF}       & 992        & 905.7    & 0.0052 \\ \hline
\textbf{MV}       & 3,706      & 269.8    & 0.0447 \\ \hline
\textbf{JE}       & 140        & 12,349.6 & 0.2435 \\ \hline
\textbf{KDD}      & 255,170    & 10.8     & \textbf{0.0000059} \\ \hline
\textbf{DG}       & 3,553      & 847.5    & 0.0061 \\ \hline
\textbf{AOL}      & 1,632,788  & 6.6      & \textbf{0.0000014} \\ \hline
\textbf{EP}       & 120,492    & 113.4    & \textbf{0.00015} \\ \hline
\textbf{Netflix}  & 17,770     & 5,654.6  & 0.0118 \\ \hline
\textbf{Yahoo}    & 624,961    & 410.9    & \textbf{0.00041} \\ \hline
\end{tabular}
\end{table}

We also compare the performance of G-BBC with that of G-BBC++. Our experiments show that G-BBC employs a fixed, tile-based strategy that performs well when the number of tiles is limited, and each tile contains sufficient computation. In such cases, the cost of tile construction and synchronization can be effectively reduced using shared memory. However, this approach does not scale well when the smaller side of the partition becomes large, as the number of tiles grows proportionally, introducing substantial overhead. To clearly observe this variation, we assess each dataset using $3$ structural metrics: (1) the smaller size partition ($n_{\min}$), (2) the average degree of the smaller side partition ($\bar d_{\min}$), and (3) the bipartite edge density ($\rho$). The minimal side size is defined as $n_{\min} = \min(|U|, |V|)$, the average degree is obtained as $\bar d_{\min} = |E|/n_{\min}$, and the global edge density with $\rho = |E|/(|U||V|)$, which represents the fraction of all possible edges that are present~\cite{Sariyuce2018}. All these metrics are reported in Table~\ref{tab:graph_density_decimal}.

These structural metrics explain the observed performance differences between G-BBC and G-BBC++. For example, the Netflix dataset contains approximately 100 million edges but has a relatively small minimum-side partition of size $17{,}770$, resulting in a very high average degree on that side. Consequently, each tile performs substantial computation, allowing G-BBC to reduce its overhead and complete the computation faster than G-BBC++. In contrast, datasets such as AOL, EP, and Yahoo exhibit extremely large minimum-side partitions, low average degrees, and very low densities (as shown in Table~\ref{tab:graph_density_decimal}). Although these graphs contain tens to hundreds of millions of edges, the work per tile is limited, and the cumulative tile-processing overhead becomes the dominant cost. As a consequence, G-BBC experiences notable efficiency reduction on these datasets. In contrast, G-BBC++ incorporates a dynamic scheduling and load-balancing strategy that adapts more effectively to large, sparse, and uneven partition structures. This adaptability allows G-BBC++ to avoid excessive processing of underutilized tiles and to gain better performance on datasets with large minimum-side partitions and low density. 

Overall, our extensive experiments show that G-BBC achieves superior performance when the minimum-side partition is relatively small. G-BBC++, however, maintains competitive performance in these cases while also scaling effectively to datasets with large, sparse, or highly imbalanced partitions. Consequently, G-BBC++ offers robust, consistent performance across a wide range of graph characteristics, making it a more versatile, broadly applicable GPU-based solution.

\section{Case Study}
\label{sec: casestudy}

In this section, we present a case study to evaluate the effectiveness of the proposed model. 

Drug combination therapies have become increasingly important for treating complex diseases, particularly when single-drug treatments are ineffective due to redundancy in biological pathways or the emergence of drug resistance. In such contexts, understanding how multiple drugs interact with common(shared) biological targets is crucial. These interactions can be synergistic (reinforcing effects), antagonistic (canceling effects), or compensatory (to mitigate the other). To represent this, we model the system as a signed bipartite graph, where one set of nodes represents drugs, the other set represents biological targets, and edges are labeled as activating ($+$) or inhibitory ($-$) interactions.

     In this case study, we explore the utility of balanced butterflies for capturing coherent and contradictory drug-target relationships. Although our core algorithm is designed for coherent and incoherent butterfly counting, we extend it here to also detect mixed coherent butterflies.` This network contains a total of $69537$ butterfly motifs. Each butterfly is classified into one of these categories based on the sign configuration of its interactions:
 \begin{itemize}
  \item \textbf{Coherent}: Both drugs have the same type of interaction (activation or inhibition) with each shared target.
  \item \textbf{Incoherent}: The drugs have opposite effects on each shared target.
  \item \textbf{Mixed}:  Each drug influences the target with a combination of coherent and incoherent signs.
\end{itemize}

\begin{table*}[htbp]
\centering
\caption{Sample drugs and targets in the signed human drug–target network.}
\label{tab:drug_target_acronyms}
\begin{tabular}{@{}ll l|| ll l@{}}
\toprule
\textbf{Acronym} & \textbf{Drug Name} & \textbf{DrugBank No.} & 
\textbf{Acronym } & \textbf{Target Name} & \textbf{Target ID} \\
\midrule
D1   & Amphetamine & DB00182 & T1 & Synaptic vesicular amine transporter & BE0000118 \\
D2  & 3,4-Methylenedioxymethamphetamine & DB01454 & T2 & Sodium-dependent dopamine transporter & BE0000647 \\
D3   & Dextroamphetamine & DB01576 & T3 & Trace amine-associated receptor 1 & BE0001044 \\
D4   & Tramadol & DB00193 & T4 & Mu-type opioid receptor & BE0000770 \\
D5   & Ziprasidone & DB00246 & T5 & Delta-type opioid receptor & BE0000420 \\
D6   & Paliperidone & DB01267 & T6 & Sodium-dependent noradrenaline transporter & BE0000486 \\
--    & -- & -- & T7 & Histamine H1 receptor & BE0000442 \\
--    & -- & -- & T8 & Alpha-2C adrenergic receptor & BE0000342 \\
\bottomrule
\end{tabular}
\end{table*}

To make these categories more tangible, we extract a small subset of the DH network containing a few representative drugs and targets. The signed interactions in this subset are summarized in Table~\ref{tab:drug_target_acronyms}, where solid edges denote activation and dashed edges denote inhibition, as shown in Fig.\ref{fig:drug-target}. From this subgraph, we identify one example each of a coherent, incoherent, and mixed butterfly, illustrated in Fig.~\ref{fig:fivebutterflies}. Panel~(a) illustrates a scenario where drug actions are coherent at each target, and the cycle is balanced, whereas panel~(b) shows incoherent drug actions, where each drug has the opposite effect on each target, and the cycle is balanced. Panel (c) presents a mixed case, where each drug has the target with a combination of coherent and incoherent signs across the targets. These examples clarify the structural differences between the three butterfly types and demonstrate how such patterns can help identify potential synergistic interactions or contradictions, thereby supporting the design of coherent therapeutic strategies.  While we do not use efficacy data or perform biological validation in this case study, our analysis focuses on how graph-theoretical structure reflects biologically interpretable coherence.

\begin{figure}[htbp]
\centering
\begin{tikzpicture}[
  node distance=1.5cm and 1cm,
  every node/.style={inner sep=0pt}
]

\node (D1) at (0.8,0) {\includegraphics[width=0.8cm]{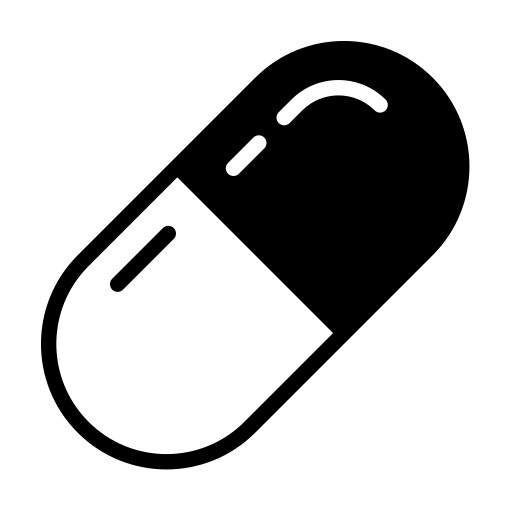}};
\node (D2) at (1.8,0) {\includegraphics[width=0.8cm]{Images/Drug.png}};
\node (D3) at (2.8,0) {\includegraphics[width=0.8cm]{Images/Drug.png}};
\node (D4) at (3.8,0) {\includegraphics[width=0.8cm]{Images/Drug.png}};
\node (D5) at (4.8,0) {\includegraphics[width=0.8cm]{Images/Drug.png}};
\node (D6) at (5.8,0) {\includegraphics[width=0.8cm]{Images/Drug.png}};

\node[above=2pt of D1] {\scriptsize D1};
\node[above=2pt of D2] {\scriptsize D2};
\node[above=2pt of D3] {\scriptsize D3};
\node[above=2pt of D4] {\scriptsize D4};
\node[above=2pt of D5] {\scriptsize D5};
\node[above=2pt of D6] {\scriptsize D6};

\node (T1) at (0,-2) {\includegraphics[width=0.8cm]{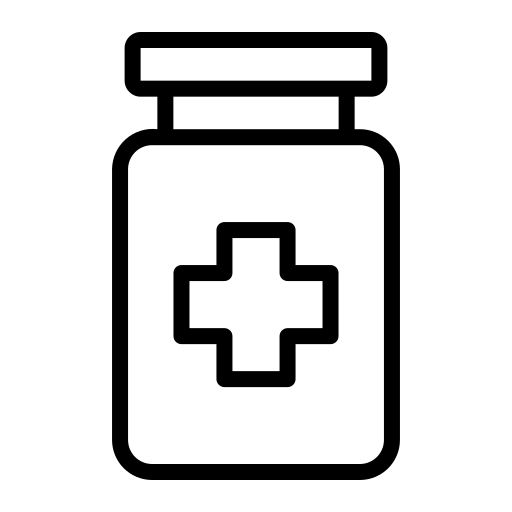}};
\node (T2) at (1.2,-2) {\includegraphics[width=0.8cm]{Images/target.png}};
\node (T3) at (2.2,-2) {\includegraphics[width=0.8cm]{Images/target.png}};
\node (T4) at (3.2,-2) {\includegraphics[width=0.8cm]{Images/target.png}};
\node (T5) at (4.2,-2) {\includegraphics[width=0.8cm]{Images/target.png}};
\node (T6) at (5.2,-2) {\includegraphics[width=0.8cm]{Images/target.png}};
\node (T7) at (6.2,-2) {\includegraphics[width=0.8cm]{Images/target.png}};
\node (T8) at (7.2,-2) {\includegraphics[width=0.8cm]{Images/target.png}};

\node[below=2pt of T1] {\scriptsize T1};
\node[below=2pt of T2] {\scriptsize T2};
\node[below=2pt of T3] {\scriptsize T3};
\node[below=2pt of T4] {\scriptsize T4};
\node[below=2pt of T5] {\scriptsize T5};
\node[below=2pt of T6] {\scriptsize T6};
\node[below=2pt of T7] {\scriptsize T7};
\node[below=2pt of T8] {\scriptsize T8};

\draw[dottedEdge]  (D1.south) -- (T1.north);
\draw[dottedEdge]  (D1.south) -- (T2.north);
\draw[solidEdge] (D1.south) -- (T3.north);
\draw[solidEdge]  (D1.south) -- (T6.north);

\draw[dottedEdge]  (D2.south) -- (T1.north);
\draw[dottedEdge] (D2.south) -- (T2.north);

\draw[dottedEdge]  (D3.south) -- (T2.north);
\draw[solidEdge]  (D3.south) -- (T3.north);
\draw[solidEdge] (D3.south) -- (T4.north);
\draw[solidEdge]  (D3.south) -- (T5.north);
\draw[dottedEdge] (D3.south) -- (T6.north);

\draw[solidEdge]  (D4.south) -- (T4.north);
\draw[solidEdge] (D4.south) -- (T5.north);
\draw[dottedEdge]  (D4.south) -- (T6.north);

\draw[solidEdge]  (D5.south) -- (T7.north);
\draw[solidEdge] (D5.south) -- (T8.north);

\draw[dottedEdge]  (D6.south) -- (T7.north);
\draw[dottedEdge]  (D6.south) -- (T8.north);

\end{tikzpicture}
\caption{Sample signed human drug–target network.}
\label{fig:drug-target}
\end{figure}

\begin{figure}[htbp]
\centering
\begin{tikzpicture}[
    node distance=0.2cm and 0.2cm,
    every node/.style={inner sep=0pt},
    edge/.style={thick},
    scale=0.85, transform shape
]

\begin{scope}[shift={(0,0)}]
    \node (d1a) at (0,0) {\includegraphics[width=0.8cm]{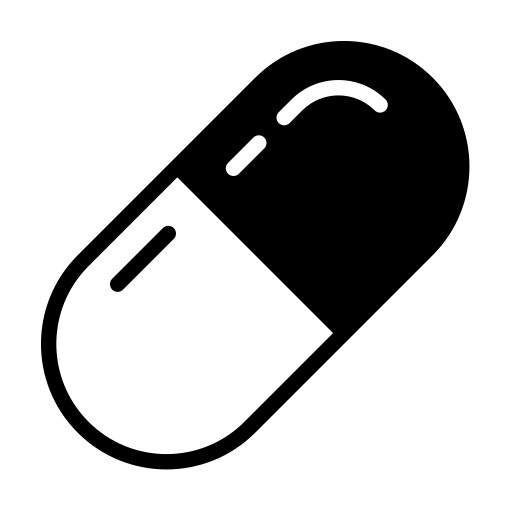}};
    \node[above=1pt of d1a] {\small D3};

    \node (d1b) at (1.6,0) {\includegraphics[width=0.8cm]{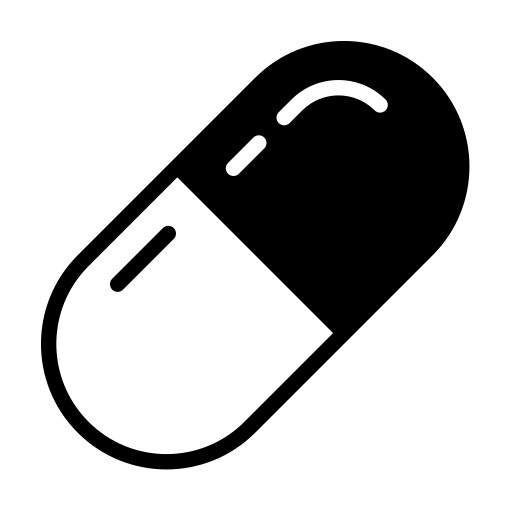}};
    \node[above=1pt of d1b] {\small D4};

    \node (t1a) at (0,-1.6) {\includegraphics[width=0.8cm]{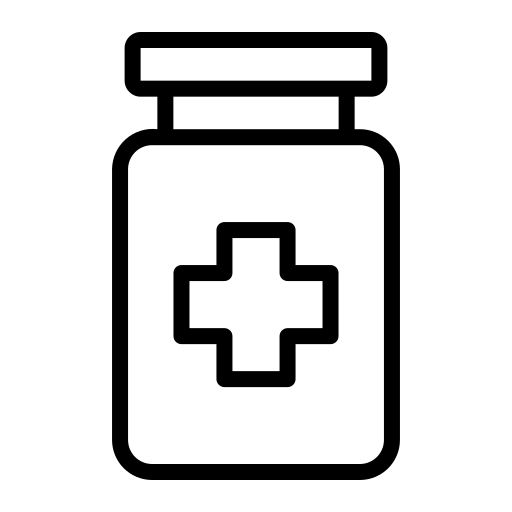}};
    \node[below=1pt of t1a] {\small T4};

    \node (t1b) at (1.6,-1.6) {\includegraphics[width=0.8cm]{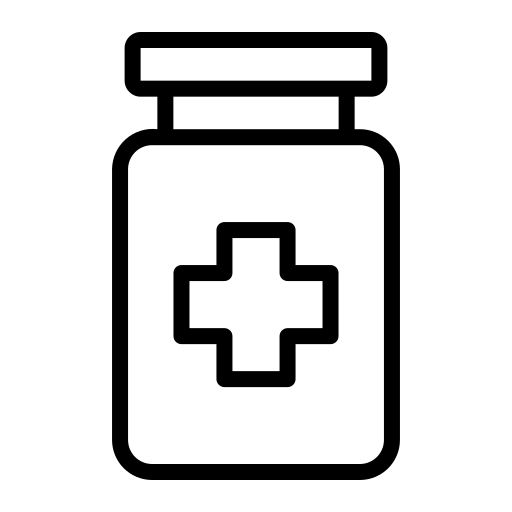}};
    \node[below=1pt of t1b] {\small T5};

    \draw[edge,solid] (d1a.south) -- (t1a.north);
    \draw[edge,solid] (d1a.south) -- (t1b.north);
    \draw[edge,solid] (d1b.south) -- (t1a.north);
    \draw[edge,solid] (d1b.south) -- (t1b.north);
\end{scope}

\begin{scope}[shift={(3,0)}]
    \node (d2a) at (0,0) {\includegraphics[width=0.8cm]{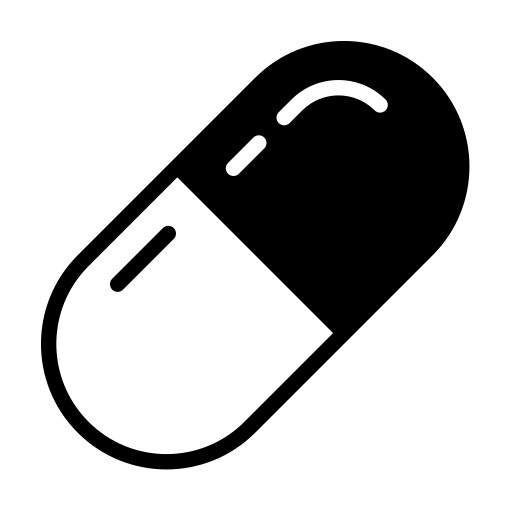}};
    \node[above=1pt of d2a] {\small D3};

    \node (d2b) at (1.6,0) {\includegraphics[width=0.8cm]{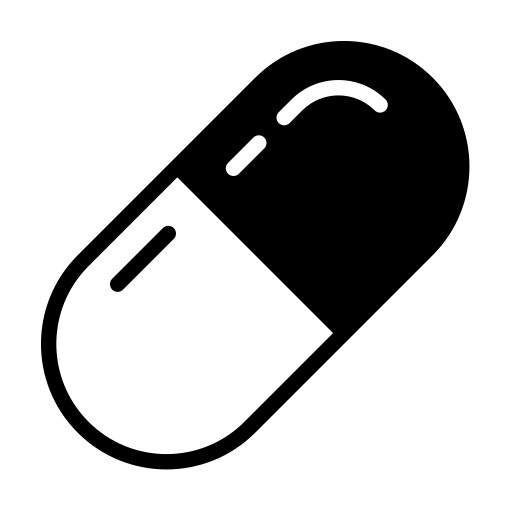}};
    \node[above=1pt of d2b] {\small D4};

    \node (t2a) at (0,-1.6) {\includegraphics[width=0.8cm]{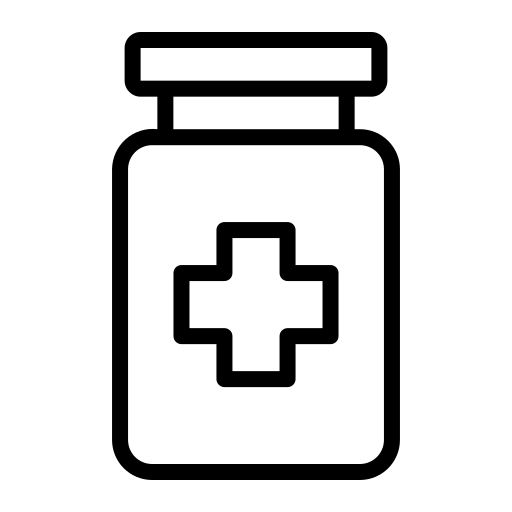}};
    \node[below=1pt of t2a] {\small T5};

    \node (t2b) at (1.6,-1.6) {\includegraphics[width=0.8cm]{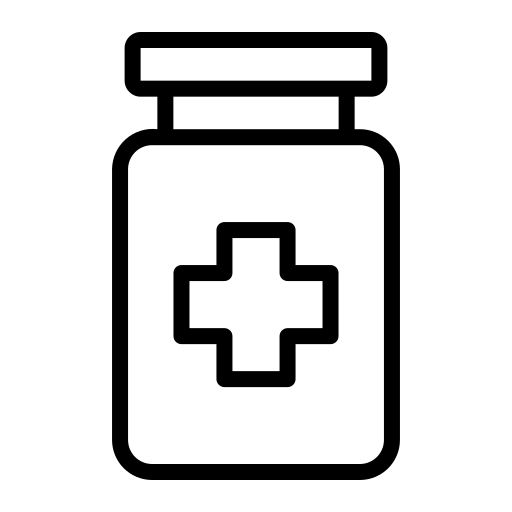}};
    \node[below=1pt of t2b] {\small T6};

    \draw[edge,solid] (d2a.south) -- (t2a.north);
    \draw[edge,solid] (d2b.south) -- (t2a.north);

    \draw[edge,dashed] (d2a.south) -- (t2b.north);
    \draw[edge,dashed] (d2b.south) -- (t2b.north);

    \node at (0.7,-2.9) {\textbf{(\large a)}};
\end{scope}

\begin{scope}[shift={(6,0)}]
    \node (d3a) at (0,0) {\includegraphics[width=0.8cm]{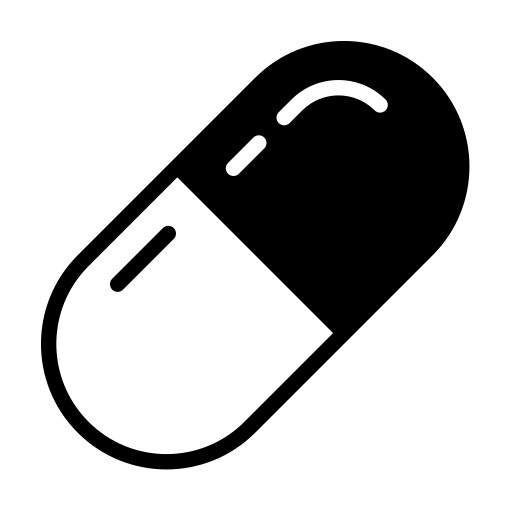}};
    \node[above=1pt of d3a] {\small D1};

    \node (d3b) at (1.6,0) {\includegraphics[width=0.8cm]{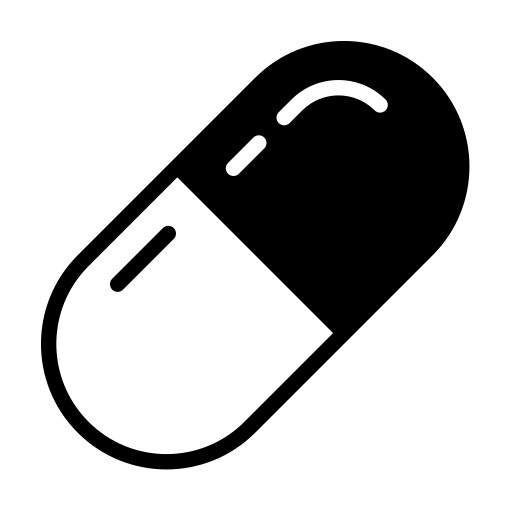}};
    \node[above=1pt of d3b] {\small D2};

    \node (t3a) at (0,-1.6) {\includegraphics[width=0.8cm]{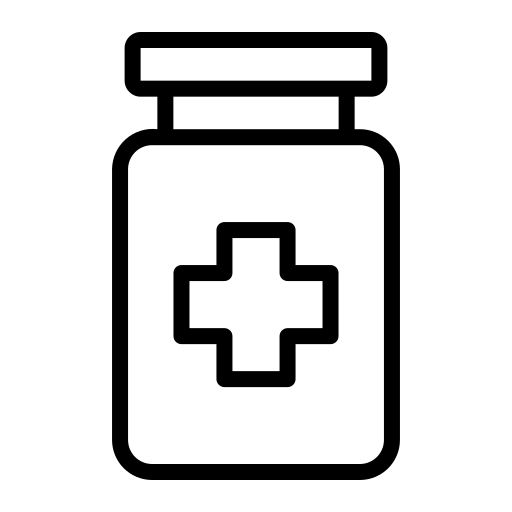}};
    \node[below=1pt of t3a] {\small T1};

    \node (t3b) at (1.6,-1.6) {\includegraphics[width=0.8cm]{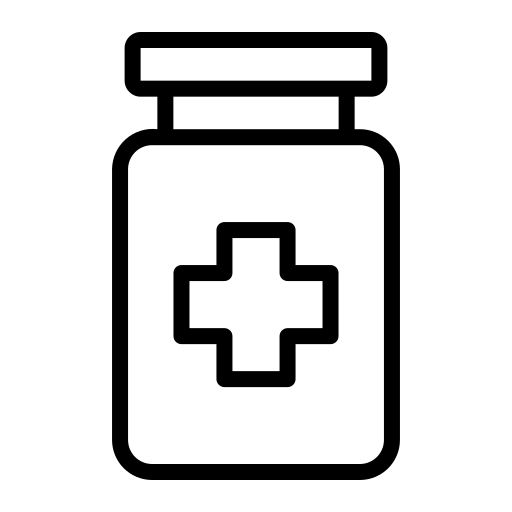}};
    \node[below=1pt of t3b] {\small T2};

    \draw[edge,dashed] (d3a.south) -- (t3a.north);
    \draw[edge,dashed] (d3a.south) -- (t3b.north);
    \draw[edge,dashed] (d3b.south) -- (t3a.north);
    \draw[edge,dashed] (d3b.south) -- (t3b.north);
\end{scope}

\begin{scope}[shift={(1,-4.0)}]
    \node (d4a) at (0,0) {\includegraphics[width=0.8cm]{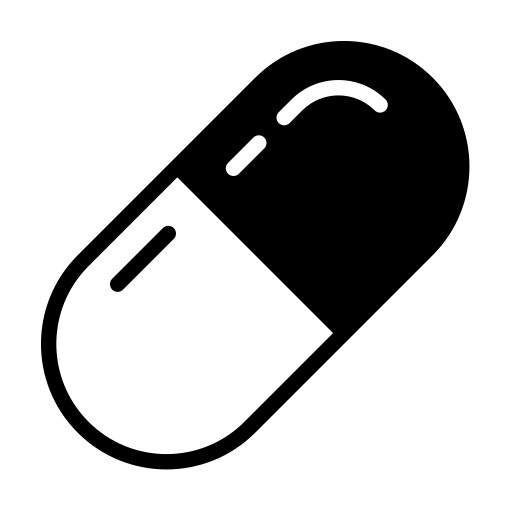}};
    \node[above=1pt of d4a] {\small D5};

    \node (d4b) at (1.6,0) {\includegraphics[width=0.8cm]{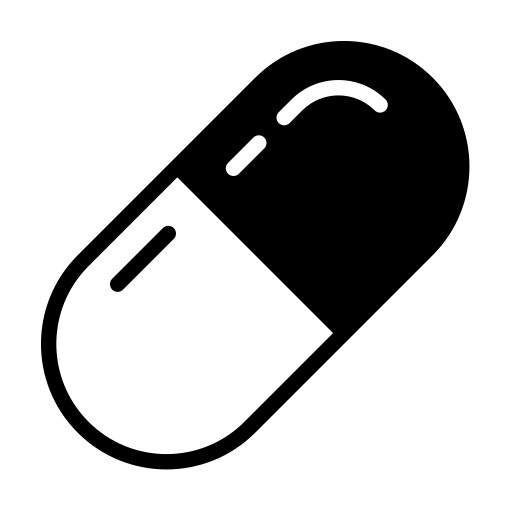}};
    \node[above=1pt of d4b] {\small D6};

    \node (t4a) at (0,-1.6) {\includegraphics[width=0.8cm]{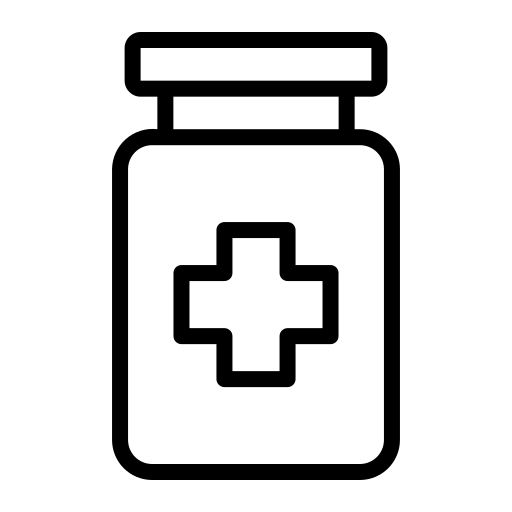}};
    \node[below=1pt of t4a] {\small T7};

    \node (t4b) at (1.6,-1.6) {\includegraphics[width=0.8cm]{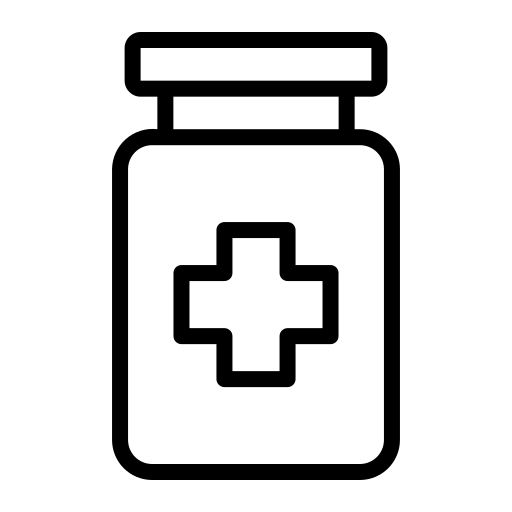}};
    \node[below=1pt of t4b] {\small T8};

    \draw[edge,solid] (d4a.south) -- (t4a.north);
    \draw[edge,solid] (d4a.south) -- (t4b.north);
    \draw[edge,dashed] (d4b.south) -- (t4a.north);
    \draw[edge,dashed] (d4b.south) -- (t4b.north);

    \node at (0.7,-2.9) {\textbf{(\large b)}};
\end{scope}

\begin{scope}[shift={(5,-4.0)}]
    \node (d5a) at (0,0) {\includegraphics[width=0.8cm]{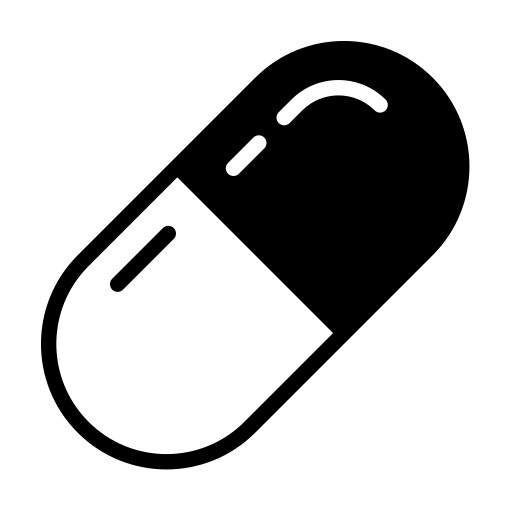}};
    \node[above=1pt of d5a] {\small D1};

    \node (d5b) at (1.6,0) {\includegraphics[width=0.8cm]{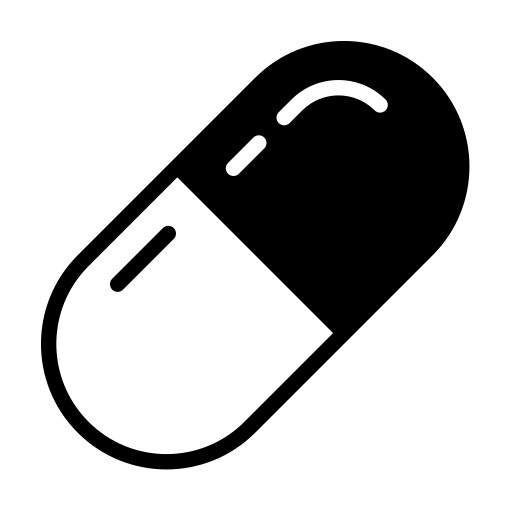}};
    \node[above=1pt of d5b] {\small D3};

    \node (t5a) at (0,-1.6) {\includegraphics[width=0.8cm]{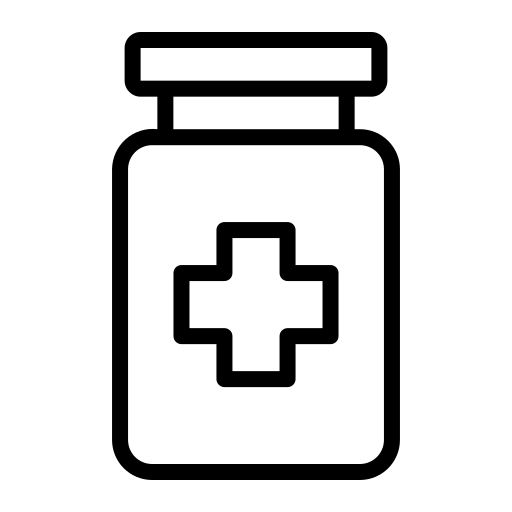}};
    \node[below=1pt of t5a] {\small T3};

    \node (t5b) at (1.6,-1.6) {\includegraphics[width=0.8cm]{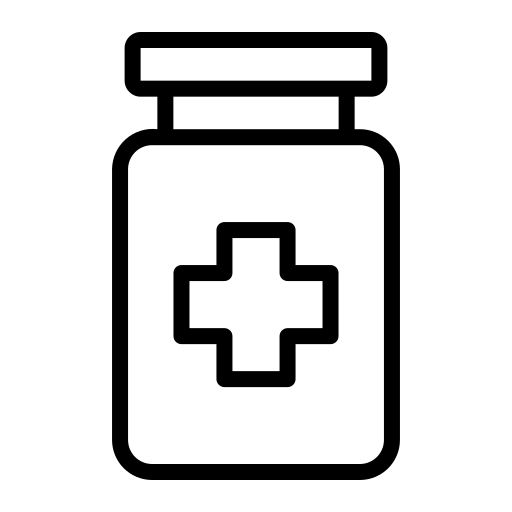}};
    \node[below=1pt of t5b] {\small T6};

    \draw[edge,solid] (d5a.south) -- (t5a.north);
    \draw[edge,solid] (d5a.south) -- (t5b.north);
    \draw[edge,solid] (d5b.south) -- (t5a.north);
    \draw[edge,dashed] (d5b.south) -- (t5b.north);

    \node at (0.7,-2.9) {\textbf{(\large c)}};
\end{scope}

\end{tikzpicture}
\caption{Panel (a) coherent, (b) incoherent, and (c) mixed butterflies in the signed human drug-target network.}
\label{fig:fivebutterflies}
\end{figure}

Applying our method to the full DH network yields the distribution shown in Table~\ref{Tab: casestudy1}, approximately $81\%$ all butterflies are coherent, suggesting a strong relationship for drug pairs to act consistently on shared targets. Incoherent butterflies make up $13.5\%$ mixed motifs, which represent conflicting and potentially unstable interactions, and are the least frequent at $5.5\%$. Importantly, coherent and incoherent butterflies are typically positive cycles (i.e., contain an even number of negative edges), while mixed butterflies are often negative cycles. This observation aligns with biological intuition that structurally balanced cycles may imply stable or reinforcing drug-target interactions, while unbalanced (negative) motifs may reflect contradictory or unstable effects.


\begin{table}[!htbp]
\small
\caption{
Coherent, incoherent, and mixed butterflies.}
\begin{center}
\begin{tabular}{|c|c|c|}
\hline
\textbf{Type of butterflies} &\textbf{Count} \\
\hline
coherent,($++ / ++$)     & $25158$ \\ \hline
coherent,($++ / --$)     & $3042$ \\ \hline
coherent,($-- / --$)     & $28166$   \\ \hline
incoherent,$(+-/+-)$     & $9375$  \\ \hline \hline
mixed,$(++/+-)$     & $1138$  \\ \hline
mixed,$(+-/--)$     & $2658$  \\ \hline

\end{tabular}
\end{center}
\label{Tab: casestudy1}
\end{table}


\section{Concluding Remarks}
\label{Sec:Con}
In this work, we studied the problem of balanced butterfly counting in signed bipartite graphs and presented the first highly parallel solutions for both multi-core CPUs and GPUs. The proposed multi-core algorithm, M-BBC, utilizes fine-grained vertex-level parallelism to accelerate wedge-based aggregation while preventing the formation of unbalanced substructures. To further improve scalability, we also proposed two GPU-based algorithms, G-BBC and G-BBC++, which utilize tile-based shared-memory processing and dynamic workload scheduling to effectively exploit the massive parallel processing capabilities of modern GPUs.  We evaluated our methods on $15$ real-world bipartite datasets, and the experimental results clearly demonstrate their efficiency and scalability. Compared to the sequential baseline, M-BBC achieves speedups of up to $71.13\times$, with an average improvement of $38.13\times$. The GPU implementation provides even greater acceleration. G-BBC and G-BBC++ achieve speedups of up to $13{,}320\times$ on average $2{,}600\times$. These results confirm that our parallel designs are highly effective and well-suited for analyzing large-scale signed bipartite graphs.


%

\section{Future Directions}
\label{Sec:fut}

In this work, we focus on balanced butterfly counting, i.e., balanced (2,2)-bicliques, as they form the simplest and most fundamental signed motif, allowing efficient processing by always considering the smaller side of the bipartite graph. Although the same ideas naturally extend to \textbf{balanced $(2,k)$-bicliques}, doing so requires fixing one side of the graph during counting. When this fixed side is the larger partition, the computational cost increases noticeably, posing practical challenges to scalability and load balancing, even though the core algorithmic ideas remain unchanged. Exploring efficient ways to handle this setting is an interesting direction for future work. More broadly, we plan to generalize the proposed techniques to balanced $(p,q)$-bicliques, allowing arbitrary values of $p$ and $q$, thereby enabling the discovery of richer, higher-order motifs in signed bipartite networks, which helps to do fraud detection, anomaly detection, and trust modeling in social networks. Another possible direction for future work is to extend this research to support dynamic signed bipartite graphs, in which edges can be added, removed, or their signs updated over time. Extending the work to these dynamic settings would allow deeper insights as the underlying networks change.

\EOD

\end{document}